\documentclass[conference]{IEEEtran}
\IEEEoverridecommandlockouts
\usepackage{cite}
\usepackage[hyphens]{url}
\usepackage[hidelinks]{hyperref}
\usepackage{graphicx}
\usepackage{subcaption}
\usepackage{soul}

\usepackage{booktabs}
\usepackage{textcomp}
\usepackage{xcolor}
\usepackage{tikz}
\usepackage{siunitx}
\usepackage{array}
\usepackage{multirow}
\usepackage{longtable}
\usepackage{adjustbox}
\usepackage{eso-pic}

\usepackage{amsmath,amsfonts,amsthm,amssymb}
\usepackage{enumitem}
\usepackage{algorithm}
\usepackage[noend]{algpseudocode}

\usepackage[most]{tcolorbox} 

\definecolor{commentbg}{RGB}{245,245,245}
\definecolor{pointcolor}{RGB}{0,0,128}

\newtcolorbox{reviewcomment}[1][]{
    colback=commentbg,
    colframe=commentbg,
    boxrule=0pt,
    arc=3pt,
    left=10pt, right=10pt, top=10pt, bottom=10pt,
    breakable,
    fonttitle=\bfseries\color{black},
    #1
}

\theoremstyle{plain}
\newtheorem{thm}{Theorem}

\newtheorem{lem}[thm]{Lemma}
\newtheorem{cor}{Corollary}[thm]

\theoremstyle{definition}

\newcommand{\IEEEArxivCopyrightNotice}{%
  \textcopyright{} 2026 IEEE.
  Personal use of this material is permitted.
  Permission from IEEE must be obtained for all other uses,
  in any current or future media, including reprinting/republishing
  this material for advertising or promotional purposes,
  creating new collective works, for resale or redistribution
  to servers or lists, or reuse of any copyrighted component
  of this work in other works.
}

\def\BibTeX{{\rm B\kern-.05em{\sc i\kern-.025em b}\kern-.08em
    T\kern-.1667em\lower.7ex\hbox{E}\kern-.125emX}}

\begin{document}

\AddToShipoutPictureFG*{%
  \AtPageLowerLeft{%
    \raisebox{5mm}{%
      \makebox[\paperwidth][c]{%
        \parbox{0.94\paperwidth}{%
          \centering
          \fontsize{6}{7}\selectfont
          \IEEEArxivCopyrightNotice
        }%
      }%
    }%
  }%
}


\twocolumn
\setcounter{page}{1}

\title{ABC: Numerical Data Collection under \\ Local Differential Privacy without Prior Knowledge
\thanks{This research was supported by Basic Science Research Program through the National Research Foundation of Korea(NRF) funded by the Ministry of Education(RS-2021-NR060143, RS-2026-25494942) and the Institute of Information \& Communications Technology Planning \& Evaluation(IITP)-ICT Creative Consilience Program grant funded by the Korea government(MSIT)(IITP-2026-RS-2020-II201819)}}
\author{\IEEEauthorblockN{1\textsuperscript{st} Incheol Baek}
\IEEEauthorblockA{\textit{Computer Science and Engineering} \\
\textit{Korea University}\\
Seoul, Republic of Korea \\
inch307@korea.ac.kr}
\and
\IEEEauthorblockN{2\textsuperscript{nd} Hyungbin Kim}
\IEEEauthorblockA{\textit{Computer Science and Engineering} \\
\textit{Korea University}\\
Seoul, Republic of Korea \\
hyungbinkim@korea.ac.kr}
\and
\IEEEauthorblockN{3\textsuperscript{rd} Yon Dohn Chung}
\IEEEauthorblockA{\textit{Computer Science and Engineering} \\
\textit{Korea University}\\
Seoul, Republic of Korea \\
ydchung@korea.ac.kr}
}

\maketitle

\AddToShipoutPictureFG*{%
  \AtPageUpperLeft{%
    \raisebox{-10mm}[0pt][0pt]{%
      \makebox[\paperwidth][c]{%
        \footnotesize\itshape
        Accepted at the 42nd IEEE International Conference on Data Engineering (ICDE 2026).%
      }%
    }%
  }%
}

\begin{abstract}
Local Differential Privacy (LDP) provides strong privacy guarantees for collecting numerical data. A fundamental challenge, however, is that existing LDP mechanisms require a predefined data domain, which is often unknown in practice. This lack of prior knowledge creates a critical dilemma for the data collector: if the chosen domain is too narrow, values outside the range are clipped, leading to information loss. Conversely, if the domain is too wide, excessive noise is added during the privatization process, which degrades the quality of collected data. This highlights the need for methods that can dynamically estimate the data domain.

In this work, we propose an adaptive LDP framework that addresses this problem. In our method, each user sends two pieces of information: their perturbed numerical data, and a privatized signal indicating if their original value was clipped by the current domain. By aggregating these signals, our proposed method, \textit{Adaptive Bounding of Clipping regions (ABC) method}, iteratively adjusts the domain to fit the underlying data distribution without prior knowledge. Our theoretical analysis shows that the estimated data domain converges to an appropriate range.

In the empirical evaluation, the results demonstrate that our framework significantly improves the quality of numerical data collection across various datasets and underlying LDP mechanisms. We also show that the estimated range successfully converges in practice and our approach is robust to its hyperparameters through comprehensive ablation studies.
\end{abstract}



\section{Introduction}

In this digital era, the spread of connected devices, from personal computers and smartphones to a vast ecosystem of Internet of Things (IoT) devices, generates large volume of data at a rapid pace. Such data hold valuable insights into device states \cite{okmi2023mobile, blondel2015survey, zinman2020utilizing}, user behaviors \cite{paul2016smartbuddy, ahmad2016defining, jalali2016human}, and sensor measurements \cite{djedouboum2018big, wang2010networked, kerle2008real}. However, the direct collection of such raw data poses significant privacy risks by exposing sensitive information about individuals \cite{wang2020comprehensive, xiong2020comprehensive}.

Local Differential Privacy (LDP) \cite{kasiviswanathan2011can} has emerged as a promising solution to this privacy challenge. LDP has been widely adopted by major technology companies, such as Apple \cite{Apple}, Microsoft \cite{ding2017collecting}, and Google \cite{erlingsson2014rappor}, for privacy-preserving data collection. Under LDP, each user perturbs their data locally before transmission to a central server. This ensures that the server never accesses raw, sensitive information, while still allowing for the estimation of aggregate statistics from the collected data. The level of this privacy-preserving noise is controlled by a parameter, $\epsilon$, known as the privacy budget. A smaller $\epsilon$ provides stronger privacy guarantees but reduces data utility, while a larger $\epsilon$ yields better utility with weaker privacy.

The functions that perform this perturbation are known as LDP mechanisms. While these mechanisms have been developed for various data types \cite{bassily2015local, wang2019consistent, wang2019locally, fanti2016building, qin2016heavy, nguyen2016collecting}, this paper focuses on numerical data \cite{duchi2018minimax, baek2025noutput, zhao2020local, wang2019collecting}, which is important for many applications. Existing mechanisms for numerical LDP are designed to perturb a user's value while preserving its statistical utility, allowing a server to collect these perturbed values and to employ them in data analysis. 

However, they assume that the data domain is known in advance. This assumption is often unrealistic in practice, where a data collector often has very little or no prior information about the range of data. This lack of knowledge creates a fundamental dilemma. A domain set too narrowly results in data clipping, which introduces systematic bias and information loss. Conversely, a domain set too widely forces the addition of excessive noise to protect the entire range, which inflates the variance of the final estimate and degrades utility.

This critical dependence on prior knowledge severely limits the practical applicability of existing LDP mechanisms. For numerical LDP to be a truly viable and reliable tool in real-world settings, methods that can operate without this assumption are not just beneficial, but essential. To address this problem, we propose an adaptive mechanism, which we call the \textbf{Adaptive Bounding of Clipping regions (ABC)} method, that adaptively bounds values using second-order optimization. In our protocol, each user sends their numerically perturbed data and a privatized signal indicating if their original value was clipped by the current domain. By aggregating these signals, the server iteratively adjusts the domain to fit the underlying data distribution. Since our method deals solely with domain adjustment, it can be seamlessly integrated with existing LDP mechanisms.

Our key contributions are summarized as follows:
\begin{itemize}[leftmargin=*]
    \item We propose a novel adaptive framework, the Adaptive Bounding of Clipping regions (ABC), for numerical data collection under LDP. A key advantage of our framework is its ability to operate effectively without requiring prior knowledge of the data domain.
    \item We provide a rigorous theoretical analysis of our method, including a Lyapunov-based proof guaranteeing the convergence of our adaptive update function.
    \item We conduct experiments on various real-world and synthetic datasets. The results demonstrate that our ABC framework significantly outperforms existing methods in terms of data utility and exhibits robustness to hyperparameter settings.
\end{itemize}


\section{Related Work}

Our work addresses the utility limitations of existing LDP mechanisms for numerical data collection by introducing a novel adaptive method. To provide context for our contribution, we first review the foundational and advanced LDP mechanisms for numerical data.

\subsection{Existing LDP Mechanisms}
There are various LDP mechanisms for numerical data, which can be characterized by their output space, $\Omega$. The output space is the set of all possible values a mechanism can produce after perturbation. The size of this space, $|\Omega|$, is a key factor that influences the utility of the mechanism. Typically, most of existing LDP mechanisms focus on mean estimation.

Duchi's mechanism \cite{duchi2018minimax} is a pioneering approach for continuous numerical data. It normalizes a user's value to a [-1, 1] range and maps it to one of two possible outputs with a probability dependent on the original value. While foundational, its utility can be limited by its binary output structure, where $|\Omega|=2$.

Three-output mechanism \cite{zhao2020local}, in contrast, was designed to improve utility. By expanding the output space to $|\Omega|=3$ with an additional output, it reduces estimation variance. This reduction is especially beneficial in small privacy budget ($\epsilon$) settings.

Piecewise mechanism \cite{wang2019collecting} (PM) is an advanced method that partitions the input domain [-1, 1] into several segments and applies non-uniform noise. It allocates a higher probability for the perturbed output to fall into the same segment as the original input. Since the output is drawn from a continuous range, its output space is infinite, $|\Omega|=\infty$.

A generalized form of the PM was proposed by Zhao et al.~\cite{zhao2020local}. This form can be configured to instantiate specific variants, such as the original PM or a Suboptimal PM (PM-SUB). While a theoretically optimal version of PM exists, its implementation is complex. In practice, these suboptimal versions are widely used as their empirical performance is nearly identical to that of the optimal version.

N-output mechanism \cite{baek2025noutput} is a generalized LDP mechanism that maps an input value to one of $N$ discrete outputs, defining an output space of $|\Omega|=N$. This framework formulates the mechanism's design as a solvable optimization problem, allowing it to be optimized for any arbitrary, finite number of outputs. Consequently, it has been shown to achieve high utility for various estimation tasks.

\textbf{Distribution Estimation.} 
Another work focuses on estimating the entire data distribution under LDP \cite{li2020estimating, baek2025noutput}. These mechanisms collect numerical data and estimate fine-grained distribution employing Expectation Maximization algorithm \cite{dempster1977maximum}. However, these mechanisms also suppose a known data domain, meaning our ABC framework could similarly serve as a complementary enhancement for them.

\subsection{Bounding methods} \label{sec:related_work_bounding}
\textbf{Adaptive Clipping in DP.} Adaptive clipping has been extensively studied in Central Differential Privacy and Federated Learning to balance bias and noise. Existing methods like DPLAC \cite{andrew2021differentially} and DCSGD \cite{wei2025dc} dynamically adjust clipping thresholds based on quantile estimation or gradient norm distributions. However, these methods typically rely on the server having access to securely aggregated statistics to compute these thresholds. In a strict LDP setting, where the server observes only individually perturbed noisy values, these reliable indicators are unavailable. Consequently, directly applying these CDP-based strategies to LDP often results in suboptimal convergence, highlighting the need for a dedicated LDP-specific solution.

\textbf{Clipping in LDP.} TOPL \cite{wang2021continuous} employs a two-phase sample splitting strategy. In the first phase, it utilizes a portion of the user to estimate the distribution of the data's magnitude, analytically deriving a clipping threshold $\theta$ that minimizes the expected MSE. This fixed threshold is then applied to the remaining users for the primary estimation. However, ToPL relies on a one-time estimation; if the initial density approximation is inaccurate due to LDP noise, the chosen bounds remain suboptimal for the entire process without any opportunity for correction.

\textbf{Our position.} A shared, critical limitation of these mechanisms is their reliance on the assumption that the data domain $[l, r]$ is known in advance. This is often unrealistic in real-world applications, as the data domain is rarely known. An incorrect choice of bounds leads to severe information loss due to excessive clipping (if the range is too narrow) or wasted privacy budget and high variance (if the range is too wide), both of which significantly degrade data utility.

Our work is to directly address this practical, yet largely overlooked, gap in the LDP literature. We introduce an iterative framework that dynamically discovers appropriate data bounds under LDP. Therefore, our method is not a replacement for existing LDP mechanisms but rather an orthogonal and complementary enhancement. It functions as a foundational layer that enables LDP mechanisms like PM or N-output to operate effectively in realistic settings where the domain is unknown. By doing so, our work bridges the gap between the numerical LDP mechanisms and their practical application, significantly improving the robustness and utility of private data collection.
\section{Preliminaries}

In this section, we introduce the concept of LDP and the method for collecting numerical data in an LDP environment.

\subsection{Local Differential Privacy}
LDP provides a strong privacy guarantee by having each user randomize their data locally before transmission. Unlike centralized differential privacy \cite{dwork2006differential}, which requires a trusted server to hold original data, individuals' data remain on their own devices. This ensures that the server or any external party never accesses an individual's sensitive raw data.

The privacy guarantee level of an LDP mechanism is quantified by a non-negative parameter, the privacy budget $\epsilon$. A smaller $\epsilon$ value provides stronger privacy by making the outputs for different inputs harder to distinguish, but it adds more noise, which reduces data utility. Conversely, a larger $\epsilon$ value provides better utility with a weaker privacy guarantee. Formally, a randomized mechanism $\mathcal{M}$ satisfies $\epsilon$-LDP if, for any two input values $v_1$, $v_2$ and for any possible output $y$, the following inequality holds:

\begin{equation}
\Pr[\mathcal{M}(v_1) = y] \le e^\epsilon \cdot \Pr[\mathcal{M}(v_2) = y]
\end{equation}

This definition ensures user privacy by making it difficult for an adversary to infer whether the original input was $v_1$ or $v_2$, even after observing the output $y$.

\subsection{Numerical Data Collection in LDP}
For clarity, our discussion focuses on the one-dimensional case following the existing LDP mechanisms \cite{duchi2018minimax, zhao2020local, wang2019collecting, baek2025noutput}. The standard framework for collecting numerical data under LDP is a three-step process performed by each user: clipping, normalization, and perturbation. First, a user's raw numerical value $v$ is clipped to a predefined public range $[l,r]$. Any value smaller than $l$ is set to $l$, and any value larger than $r$ is set to $r$.
\begin{equation} \label{eq:clip}
v_{\text{clip}} = \max(l, \min(r, v))
\end{equation}
Second, the clipped value $v_{\text{clip}}$ is normalized to a standardized domain, typically $[-1,1]$.
\begin{equation}
v' = \frac{2(v_{\text{clip}} - l)}{r-l} - 1
\end{equation}
Finally, a randomized LDP mechanism $\mathcal{M}$ is applied to the normalized value $v'$ to generate a perturbed output $\tilde{v}$. This perturbed value is then sent to the data collector.
\begin{equation}
\tilde{v} = \mathcal{M}(v')
\end{equation}
The mechanism $\mathcal{M}: [-1, 1] \to Y$ is a randomized function that satisfies $\epsilon$-LDP, where $Y$ is the output space. LDP mechanisms allow the data collector to construct an unbiased estimator for the original statistics. The collector aggregates the perturbed outputs $\tilde{v}$ from numerous users and, knowing the public range $[l,r]$, performs an inverse transformation to compute statistics like the mean. While each individual value is noisy, the aggregation process can recover the true statistics with some error.

While this paper focuses on the univariate case, the proposed method can be extended to high-dimensional settings using dimension-sampling \cite{zhao2020local, wang2019collecting, baek2025noutput}. In this technique, each user privatizes and reports a value for only one randomly selected dimension, thereby conserving the total privacy budget.

\subsection{Generalized Randomized Response}
While there are numerous LDP mechanisms for categorical data collection \cite{wang2017locally, erlingsson2014rappor,warner1965randomized}, a standard LDP mechanism for $d$-dimension categorical data is Generalized Randomized Response (GRR) \cite{wang2017locally}. In GRR, a user with a true category $c$ reports it with probability $p$, and reports a random category (uniformly chosen from all $d$ categories) with probability $1-p$. The probabilities $p$ and $q$ are defined as:
\begin{equation}
    p = \frac{e^\epsilon}{e^\epsilon + d - 1} \quad \text{and} \quad q = \frac{1}{e^\epsilon + d - 1}
\end{equation}
After collecting $N$ perturbed reports, the server obtains a noisy count for each category, denoted as $N'_c$. The server then applies the GRR correction formula to obtain an unbiased estimate of the true count, $\hat{N}_c$:
\begin{equation} \label{eq:grr_correction}
    \hat{N}_c = \frac{N'_c - N \cdot q}{p - q}
\end{equation}
Finally, the estimated ratio for each category, $\hat{\theta}_c$, is computed by normalizing this count: $\hat{\theta}_c = \hat{N}_c / N$.



\section{Problem Formulation and Analysis} \label{sec:problem_definition}

When the server collects numerical data $v$ from the clients, the primary challenge in the standard LDP framework is the assumption that a suitable data domain $[l, r]$ is known in advance. However, in many real-world scenarios, the server has no prior knowledge of the data's true range, which we denote as $[l^*, r^*]$. Thus, the server may guess the appropriate data domain,  but if the chosen range is too narrow, a significant fraction of data is clipped, resulting in a biased estimate. Conversely, if the range is too wide, the noise added by the LDP mechanism becomes excessive relative to the signal, which inflates the variance of the estimate.

The problem we address is: How can a server tune the bounds $[l, r]$ to converge towards a proper range that minimizes the total estimation error, without prior knowledge?

To motivate our approach, we first quantitatively analyze the two opposing sources of error an adaptive algorithm must balance.The performance of numerical LDP mechanisms is commonly evaluated by the error in mean estimation, measured using the Mean Squared Error (MSE). The MSE quantifies the total error of an estimator, defined as the expected squared difference between the estimated mean $\hat{\mu}$ of the perturbed data, and the true mean $\mu$ of the original data. The MSE of an estimator $\hat{\mu}$ is defined as $\text{MSE}(\hat{\mu}) = \mathbb{E}[(\hat{\mu} - \mu)^2]$, which can be decomposed into two components:
\begin{equation} \label{eq:mse}
    \text{MSE}(\hat{\mu}) = \text{Bias}(\hat{\mu})^2 + \text{Var}(\hat{\mu})
\end{equation} 
Bias represents the systematic error of the estimator (how far the average estimate is from the true value), while variance represents the random error (how much the estimates fluctuate).

\begin{lem}[Clipping-induced bias] \label{lem:bias}
The MSE of the estimator $\hat{\mu}$ is lower-bounded by the squared bias, a function of the out-of-bounds probability mass $P_{\text{out}}$ and magnitude $\delta_{\text{out}}$:
$$\text{MSE}(\hat{\mu}) \ge \text{Bias}(\hat{\mu})^2 = \Omega((P_{\text{out}} \cdot \delta_{\text{out}})^2)$$
\end{lem}
\begin{proof}
Since variance is non-negative, $\text{Var}(\hat{\mu}) \ge 0$, it directly follows that $\text{MSE}(\hat{\mu}) \ge \text{Bias}(\hat{\mu})^2$.
The bias is the difference between the expected value of the estimator and the true mean, $\text{Bias}(\hat{\mu}) = \mathbb{E}[\hat{\mu}] - \mu$. An unbiased LDP mechanism provides an estimate for the mean of the pre-perturbation data, which in our case is the clipped data, $v_{\text{clip}}$. Thus, $\mathbb{E}[\hat{\mu}] = \mathbb{E}[v_{\text{clip}}]$. The bias is therefore:
\begin{equation}
    \text{Bias}(\hat{\mu}) = \mathbb{E}[v_{\text{clip}}] - \mathbb{E}[v] = \mathbb{E}[v_{\text{clip}} - v]
\end{equation}
This expectation is non-zero only for values outside the range $[l, r]$. Let $p(v)$ be the probability density function of the data. The bias can be expressed as an integral over the out-of-bounds regions:
\begin{equation} \label{eq:bias}
    \text{Bias}(\hat{\mu}) = \int_{-\infty}^{l} (l - v) p(v) \,dv + \int_{r}^{\infty} (r - v) p(v) \,dv
\end{equation}
This expression is proportional to the probability mass of the clipped data ($P_{\text{out}}$) and the average magnitude of the clipping ($\delta_{\text{out}}$), which concludes the proof.
\end{proof}
Lemma \ref{lem:bias} reveals that any data falling outside the range $[l, r]$ introduces a systematic bias, which establishes a hard floor on the achievable MSE. The term $P_{\text{out}} \cdot \delta_{\text{out}}$ represents the total impact of clipping-induced error. This highlights that the overall bias is a product of both the proportion of clipped data ($P_{\text{out}}$) and its average distance from the boundary ($\delta_{\text{out}}$). This bias pulls the final estimate away from the true mean, an error that cannot be reduced by simply collecting more data. Consequently, this implies widening its bounds minimizes clipping and reduces this systematic error.

\begin{lem}[Range-induced variance inflation] \label{lem:variance}
The variance of the mean estimator $\hat{\mu}$ is quadratically proportional to the chosen range width $W=r-l$. Compared to the variance of an optimal estimator $\hat{\mu}^*$ obtained with an optimal width $W^*$, the inflation factor is:
$$\frac{\text{Var}(\hat{\mu})}{\text{Var}(\hat{\mu}^*)} = \left(\frac{W}{W^*}\right)^2$$
\end{lem}
\begin{proof}
Let $\mathcal{M}$ be a normalized LDP mechanism on $[-1,1]$ with inherent variance $\sigma_{\mathcal{M}}^2$. For $n$ users, the server computes the average of the perturbed normalized values, $\bar{v}' = \frac{1}{n} \sum_{i=1}^{n} \tilde{v}'_i$. The variance of this average is:
\begin{equation}
    \text{Var}(\bar{v}') = \frac{\text{Var}(\tilde{v}'_i)}{n} = \frac{\sigma_{\mathcal{M}}^2}{n}
\end{equation}
The server then de-normalizes this average to obtain the final estimate $\hat{\mu}$. Let $W = r-l$. The inverse transformation is a linear function:
\begin{equation}
    \hat{\mu} = \left(\frac{\bar{v}' + 1}{2}\right)W + l = \left(\frac{W}{2}\right)\bar{v}' + \left(\frac{W}{2} + l\right)
\end{equation}
Using the property of variance, $\text{Var}(aX+b) = a^2\text{Var}(X)$, we derive the variance of the final estimator $\hat{\mu}$:
\begin{align}
    \text{Var}(\hat{\mu}) &= \text{Var}\left( \left(\frac{W}{2}\right)\bar{v}' + \left(\frac{W}{2} + l\right) \right) \\
    &= \left(\frac{W}{2}\right)^2 \text{Var}(\bar{v}') \\
    &= \frac{W^2}{4} \cdot \frac{\sigma_{\mathcal{M}}^2}{n} = \frac{W^2 \sigma_{\mathcal{M}}^2}{4n} \label{eq:var_final}
\end{align}
Equation \eqref{eq:var_final} shows that $\text{Var}(\hat{\mu})$ is quadratically proportional to the range width, i.e., $\text{Var}(\hat{\mu}) \propto W^2$. Let $W^*$ be the optimal range width that minimizes the MSE, and let $\hat{\mu}^*$ be the corresponding estimator. The variance for this optimal estimator is thus $\text{Var}(\hat{\mu}^*) \propto (W^*)^2$. The ratio of the variances is therefore:
\begin{equation}
    \frac{\text{Var}(\hat{\mu})}{\text{Var}(\hat{\mu}^*)} = \frac{W^2}{(W^*)^2} = \left(\frac{W}{W^*}\right)^2
\end{equation}
which concludes the proof.
\end{proof}


Lemma \ref{lem:variance} describes how an overly wide range increases the random error. The term $(W/W^*)^2$ means that if we set our range to be unnecessarily twice as wide as the optimal range, the variance increases fourfold. In our problem, this means that a wide range makes the LDP noise dominate the true data signal, making the final estimate unreliable and unstable. This implies that an effective algorithm must narrow its bounds to control this variance.

These lemmas reveal a fundamental bias-variance trade-off. Lemma \ref{lem:bias} implies a pressure to widen the range $[l,r]$ to minimize bias. Conversely, Lemma \ref{lem:variance} implies a pressure to narrow the range to control variance. An effective algorithm must therefore find an optimal balance between these competing forces. This implies that fitting $[l, r]$ to the true bounds $[l^*, r^*]$ would not be optimal.

\begin{figure*}[ht]  
    \centering      
    \includegraphics[width=0.85\textwidth]{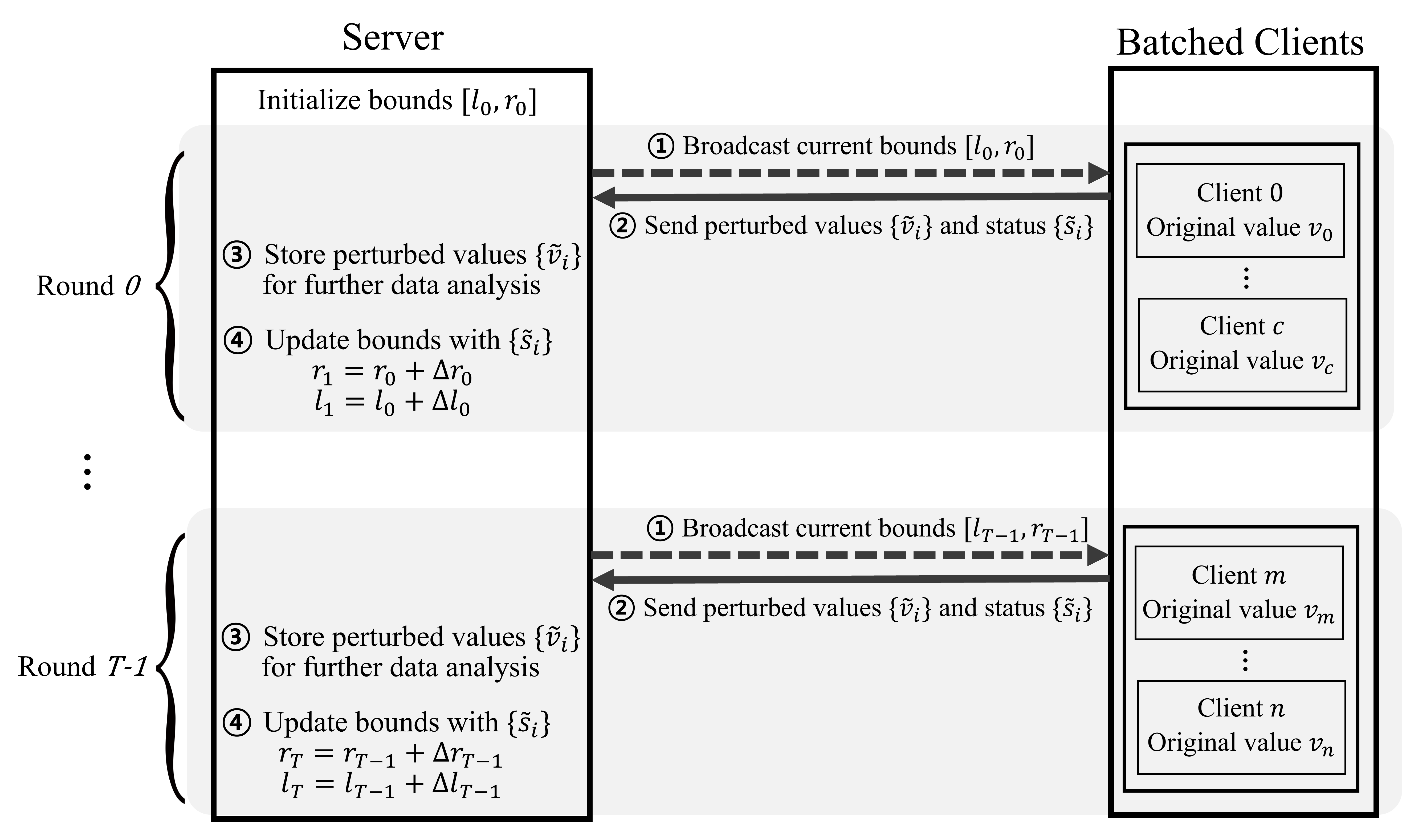}
    \caption{Overview of the ABC framework. In each round $t$, the server and a batch of clients interact to refine the data bounds. The server broadcasts the current bounds $[l_t, r_t]$, and clients respond with privatized signals ($\tilde{v}_i, \tilde{s}_i$). The server stores $\tilde{v}_i$ for further data analysis and uses $\tilde{s}_i$ to compute the next bounds $[l_{t+1}, r_{t+1}]$.}
    \label{fig:framework}
\end{figure*}

\section{Proposed Method} \label{sec:method}
This section outlines the \textbf{ABC} method to find an appropriate data domain $[l, r]$ for numerical data collection under LDP. We first present the overall framework for updating $[l, r]$ and introduce our update function. Then, we present a detailed formulation of the update function and a formal convergence analysis.

\subsection{The proposed framework: Adaptive Bounding of
Clipping regions} \label{sec:5a}
To find the appropriate data domain from a state of zero knowledge, we propose an iterative protocol that operates in rounds. At each round, clients not only provide the perturbed numerical data $\tilde{v}$, but also transmit a privatized signal $\tilde{s}$ indicating where their true value lies relative to the given bounds $[l,r]$. By aggregating these signals, the server can iteratively adjust the bounds to better fit the data distribution.

The overall framework is illustrated in Figure~\ref{fig:framework} and explained as follows.
\begin{enumerate}[leftmargin=*]
    \item \textbf{Initialization:} The server begins by setting initial bounds $[l_0, r_0]$.

    \item \textbf{Iteration for each round $t=0, 1, 2, \dots, T-1$:}
    \begin{enumerate}
        \item \textit{Server Broadcasts:} The server sends the current bounds $[l_t, r_t]$ to the batch of clients in this round.
        \item \textit{Client Reports:} Each client $i$ with a private value $v_i$ performs two actions. First, it determines its clipping status, $s_i \in \{\text{left, right, in}\}$, which indicates if $v_i < l_t$, $v_i > r_t$, or $l_t \le v_i \le r_t$, respectively. Second, it clips its value to the given range, $v_{i, \text{clip}} = \max(l_t, \min(r_t, v_i))$. The client then sends two privatized signals to the server:
        \begin{itemize}
            \item A perturbed clipping status, $\tilde{s}_i$, generated by applying an LDP mechanism for categorical data, such as GRR, to the true status $s_i$.
            \item A perturbed numerical value, $\tilde{v}_{i}$, generated by applying a numerical LDP mechanism to the clipped value $v_{i, \text{clip}}$.
        \end{itemize}
        The introduction of the status signal $\tilde{s}_i$ does not inflate the overall privacy budget. To maintain the strict total privacy budget $\epsilon$, the client splits it using a hyperparameter $\beta \in (0,1)$, allocating $(1-\beta) \epsilon$ to $\tilde{s}_i$ and $\beta \epsilon$ to $\tilde{v}_i$. Following the composition theorem, the ABC method is $\epsilon$-LDP.
        \item \textit{Server Updates:} The server aggregates the reports from batched clients in each round. Firstly, the server store the perturbed numerical values for further data analysis. Then, the server uses the set of perturbed statuses $\{\tilde{s}_i\}$ to estimate the proportions of left-clipped data ($\hat{\theta}_{l,t}$), right-clipped data ($\hat{\theta}_{r,t}$), and in-bound data ($\hat{\theta}_{0,t}$). Based on these estimated ratios, the server computes the adjustments $\Delta l_t$ and $\Delta r_t$ using our proposed update function. Finally, the bounds are updated for the next round:
        \begin{equation}
            r_{t+1} = r_t + \Delta r_t \quad \text{and} \quad l_{t+1} = l_t + \Delta l_t 
        \end{equation}
    \end{enumerate}
\end{enumerate}
This iterative process allows the server to dynamically adjust the data range based on empirical evidence provided by the clients. The ABC method leverages the estimated clipping ratios to guide the bounds toward a more optimal range. 

To implement this, the ABC method updates the bounds driving clipped data ratios ${\theta}_{r,t}$ and ${\theta}_{l,t}$ to be a target clipping ratio, $\alpha$, as visualized in Figure~\ref{fig:dist}. The update functions, whose formal derivation and convergence analysis are presented in Section \ref{sec:formulation}, are defined as follows:
\begin{align}
\Delta r_t &= \eta \cdot \left( \frac{r_t-l_t}{\max(\hat{\theta}_{0,t}, \zeta)} \right) \cdot \text{sign}(e_{r,t}) \cdot |e_{r,t}|^{\tau} \label{eq:update_r} \\
\Delta l_t &= - \eta \cdot \left( \frac{r_t-l_t}{\max(\hat{\theta}_{0,t}, \zeta)} \right) \cdot \text{sign}(e_{l,t}) \cdot |e_{l,t}|^{\tau} \label{eq:update_l}
\end{align}
where $\eta$ is a base learning rate that controls the magnitude of the bound adjustments, $\zeta$ is a small constant for numerical stability to prevent division by zero, and $\tau$ is an amplification parameter. The error terms $e_{r,t}$ and $e_{l,t}$ are defined as the difference between the estimated and target clipping ratios:
\begin{equation} \label{eq:19}
    e_{r,t} = \hat{\theta}_{r,t} - \alpha \ , \quad \quad  e_{l,t} = \hat{\theta}_{l,t} - \alpha
\end{equation}

ABC method is agnostic to the internal distribution within the domain, as the update mechanism relies exclusively on the proportion of clipped data at the boundaries.

\textbf{Efficiency and Practicality.} 
The ABC framework is highly efficient, imposing negligible communication and computational overhead. The only additional communication cost per user is for the privatized clipping status ($\tilde{s}_i$), which requires just 2 bits of data. This payload is insignificant compared to the accompanying numerical value. Computationally, the framework is also lightweight: clients perform only a simple comparison and apply a standard LDP mechanism, while the server's additional task is aggregating status counts and executing basic arithmetic update functions. This efficient design makes ABC a practical solution that can be readily integrated into existing systems without introducing significant bottlenecks. Furthermore, for high-dimensional data, the framework can be easily extended using standard dimension sampling techniques, where each user reports only one randomly selected dimension. This ensures that the communication cost per user remains constant $O(1)$ regardless of dimensionality, while the server's computational cost scales linearly $O(d)$, maintaining practicality even in high-dimensional settings.

Additionally, it is important to note that ABC is designed for scenarios involving continuous or multi-round data collection. While it requires iterative refinement, it does not mandate repeated participation from the same users. In practice, the server can partition a large user base into disjoint batches, applying updated bounds to sequential groups, thereby mitigating the limitation of single-round estimation.

In the following sections, we present derivation of the update function and convergence analysis of the update function.

\subsection{Update Function Formulation} \label{sec:formulation}
The ultimate goal is to find bounds that minimize the MSE, which decomposes into bias and variance as Eq. \eqref{eq:mse}. However, since the bias term consists of data distribution term:
\begin{equation}
    \text{Bias}(\hat{\mu}) = \int_{-\infty}^{l_t} (l_t - v) p(v) \,dv + \int_{r_t}^{\infty} (r_t - v) p(v) \,dv,
\end{equation}
the direct optimization for MSE is infeasible.

\begin{figure}[ht] 
    \centering        
    \includegraphics[width=0.45\textwidth]{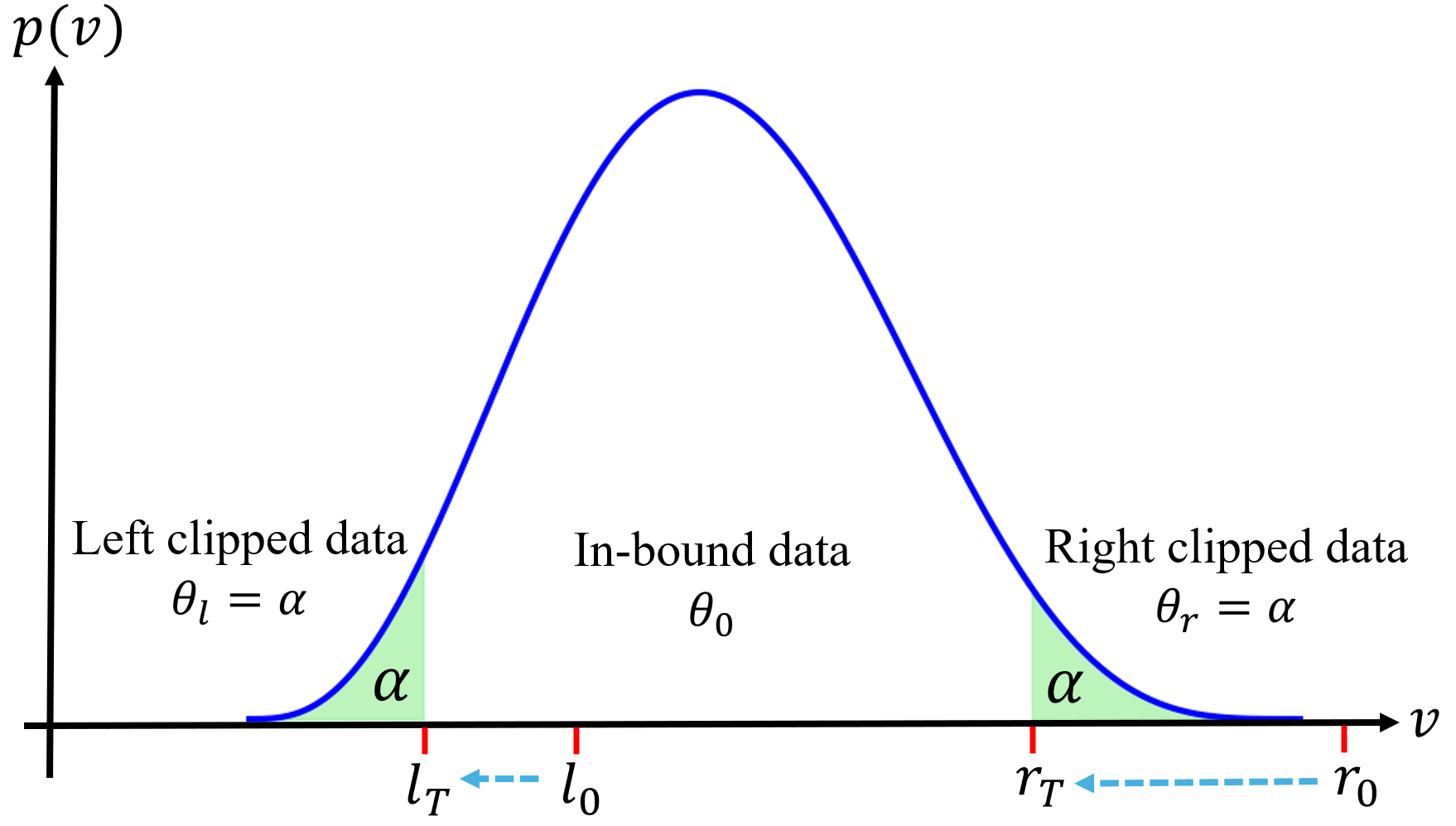}
    \caption{Illustration of the ABC method's objective. The algorithm iteratively adjusts the data bounds $[l, r]$. The objective is to find bounds where the proportion of clipped data on each side converges to a target clipping ratio, $\alpha$.}
    \label{fig:dist}
\end{figure}

\textbf{Surrogate objective.} Instead, we propose a surrogate objective: to iteratively adjust the bounds until the clipping ratio on each side, $\theta_{r,t}$ and $\theta_{l,t}$, converges to a predefined target ratio, $\alpha$. The hyperparameter $\alpha \in [0, 0.5)$ serves as a direct control mechanism for the bias-variance trade-off. As established by Lemma~\ref{lem:bias} and Lemma~\ref{lem:variance}, the choice of $\alpha$ allows a data collector to prioritize low bias (small $\alpha$) or low variance (larger $\alpha$). While a theoretically optimal $\alpha^*$ that minimizes the MSE exists for any given distribution, its value is unknown without access to that distribution. Consequently, $\alpha$ serves as a hyperparameter that provides control and flexibility over the trade-off.

To achieve this surrogate objective, we formulate the problem as minimizing a loss function employing the second-order optimization. We define our loss function, $\mathcal{L}(r)$, as a generalized form of absolute error:
\begin{equation}
    \mathcal{L}(\theta_r;\alpha) = \frac{1}{\tau_0} | \theta_r - \alpha  |^{\tau_0}
\end{equation}
where $\tau_0 > 0$ is a parameter that controls the shape of the loss function. The update function for a second-order method, such as Newton's method, is defined as:
\begin{equation} \label{eq:newton_rule} 
\Delta r = - \left(\frac{\partial^2 \mathcal{L}}{\partial r^2}\right)^{-1} \cdot \frac{\partial \mathcal{L}}{\partial r}
\end{equation}

The first derivative of this loss function respect to $r$ is:
\begin{align} \label{eq:grad}
    \frac{\partial \mathcal{L}}{\partial r} &=  \text{sign} (\theta_r - \alpha) \cdot | \theta_r - \alpha|^{\tau_0-1}  \cdot \frac{\partial}{ \partial r} \theta_r\\
    &= \text{sign} (\theta_r - \alpha) \cdot | \theta_r - \alpha|^{\tau_0-1} \cdot \frac{\partial}{ \partial r} \int_r^\infty p(v) dv  \\
    &= - p(r) \cdot \text{sign} (\theta_r - \alpha) \cdot|\theta_r - \alpha|^{\tau_0-1}
\end{align}
where $p(r)$ is the true density at the boundary $r$. While a first-order optimization approach is possible using this gradient, it lacks an inherent mechanism for dynamically adapting the learning rate, requiring careful manual tuning. In contrast, a second-order method naturally provides an adaptive learning rate by incorporating curvature information. Differentiating Eq.~\eqref{eq:grad} further yields the second derivative:
\begin{equation}
    \frac{\partial^2 \mathcal{L}}{\partial r^2} = p(r)^2 (\tau_0-1)|e_r|^{\tau_0-2} - p'(r) \text{sign}(e_r) |e_r|^{\tau_0-1}   \\
\end{equation}
where $e_r = \theta_r - \alpha$. However, the server does not know $p(r)$ and its derivative $p'(r)$. To proceed, we approximate these unknown quantities using observable values. We hypothesize that the density at a boundary, $p(r)$, can be assumed by the average density within the central region $[l, r]$. Applying this to our iterative framework, at any given round $t$, $p(r_t)$ can be defined as:
\begin{equation} \label{eq:pr}
    p(r_t) = \frac{\hat{\theta}_{0,t}}{(r_t-l_t)}
\end{equation}
Consequently, $p'(r)$ can be estimated as $\Delta p'(r_t)$. However, the objective of second-order optimization is to provide a dynamic learning rate according to the loss curvature, especially when the error is small. As $|\theta_r - \alpha|$ becomes smaller, the term $p(r)^2 (\tau_0-1)|e_r|^{\tau_0-2}$ becomes dominant. Rather than estimating $p'(r)$, we approximate $\frac{\partial^2 \mathcal{L}}{\partial r^2}$ directly as:
\begin{equation} \label{eq:second-order}
    \frac{\partial^2 \mathcal{L}}{\partial r^2} \approx  p(r)^2 (\tau_0-1)|e_r|^{\tau_0-2}
\end{equation}
This approximation is intentionally designed for high accuracy as the error approaches zero. For large errors, moving in the correct direction is more important than the step size's precise magnitude. Conversely, in the critical regime near convergence, a correct learning rate is essential to prevent overshooting and ensure stable convergence. Substituting Eq. \eqref{eq:grad}, \eqref{eq:second-order} into Eq. \eqref{eq:newton_rule}, the update function becomes:
\begin{equation} \label{eq:true_update}
    \Delta r_t = \eta \cdot \frac{(r_t-l_t)}{\hat{\theta}_{0,t}} \cdot \text{sign} (e_{r,t}) \cdot|e_{r,t}|
\end{equation}
where $e_{r,t} = \hat{\theta}_{r,t} - \alpha $. To prevent division by zero, we replace $\hat{\theta}_{0,t}$ with $\max(\hat{\theta}_{0,t}, \zeta)$ for a small constant $\zeta > 0$. Furthermore, to amplify low-signal errors (small $e_{r,t}$) and facilitate faster convergence near the target, we raise the error term to the power of $\tau$ that controls the sensitivity. Then, the final update function becomes:
\begin{align}
\Delta r_t &= \eta \cdot \left( \frac{r_t-l_t}{\max(\hat{\theta}_{0,t}, \zeta)} \right) \cdot \text{sign}(e_{r,t}) \cdot |e_{r,t}|^{\tau} \\
\Delta l_t &= - \eta \cdot \left( \frac{r_t-l_t}{\max(\hat{\theta}_{0,t}, \zeta)} \right) \cdot \text{sign}(e_{l,t}) \cdot |e_{l,t}|^{\tau} 
\end{align}
We discuss the theoretical role of $\tau$ in Section \ref{sec:convergence} and its empirical results in Section \ref{sec:exp}. We now elaborate on the key components of our update function.

\textbf{The update is scale invariance.} An update of a fixed magnitude would be significant for a narrow range but negligible for a wide one. To be robust, the algorithm's adjustments must be consistent across different data scales. This requires the absolute update $\Delta r_t$ to be proportional to the current range width $r_t - l_t$. Our update function achieves scale invariance by ensuring that $\Delta r_t \propto (r_t - l_t)$.

\textbf{Dynamic learning rate.} With a base learning rate $\eta$, the term $\eta \cdot  \frac{r_t-l_t}{\max(\hat{\theta}_{0,t}, \zeta)} $ acts as an adaptive learning rate. This structure is derived from a second-order optimization. It automatically scales the update step based on the estimated density, ensuring larger steps in larger $\hat{\theta}_{0,t}$ and smaller steps in smaller $\hat{\theta}_{0,t}$, which promotes stable convergence.

\textbf{Error response.} The theoretically derived rule Eq.~\eqref{eq:true_update} results in a linear response to the error, $|e_{r,t}|$. As the system approaches its target ($e_{r,t} \to 0$), linear updates can lead to premature convergence. Although the dynamic learning rate partially addresses this problem, we replace the linear term with a non-linear signal, $|e_{r,t}|^{\tau}$, which amplifies weak signals. This is a standard technique in control theory to ensure robustness and faster convergence near the optimum, which is proved in the following section.

\subsection{Convergence Analysis} \label{sec:convergence}
We now show that the proposed update function ensures the error converges to zero. The analyses are presented for the right bound, $r$.

\begin{lem}[Convergence] \label{lem:convergence}
For the right bound $r$, the error term $e_{r,t} = \hat{\theta}_{r,t} - \alpha$ converges to zero as $t \to \infty$, provided the learning rate $\eta$ is sufficiently small.
\end{lem}
\begin{proof}
We use Lyapunov stability analysis for the error of the right bound, $e_{r,t}$. We define a Lyapunov function $V(e_{r,t}) = \frac{1}{2} e_{r,t}^2$, which is non-negative and zero only at $e_{r,t} = 0$. The system is stable if $\Delta V_t = V(e_{r,t+1}) - V(e_{r,t}) \leq 0$.

The error evolves as $e_{r,t+1} = e_{r,t} + \Delta e_{r,t}$
The change in the Lyapunov function is:
\begin{equation}
    \Delta V_t = V(e_{r,t+1}) - V(e_{r,t}) = e_{r,t} \Delta e_{r,t} + \frac{1}{2}(\Delta e_{r,t})^2
\end{equation}
The change in error $\Delta e_{r,t}$ is induced by the bound update $\Delta r_t$. As $\Delta \theta_r = \theta_{r, t+1} - \theta_{r, t}$, $\Delta \theta_r$ can be approximated by the Taylor expansion: $\Delta \theta_r \approx \frac{\partial \theta_r}{\partial r}\bigg|_{r_t} \Delta r_t$. Since $\hat{\theta}_{r,t}$ is an estimate of ${\theta}_{r,t}$ , the change in the error also can be approximated as:
\begin{align}
    \Delta e_{r,t} &= e_{r,t+1} - e_{r,t} \\
    &= \hat{\theta}_{r,t+1}  - \hat{\theta}_{r,t} \\
    &\approx \Delta \theta_r \approx \frac{\partial \theta_r}{\partial r}\bigg|_{r_t} \Delta r_t \\
    &= -p(r_t) \Delta r_t \label{eq:delta_e}
\end{align}
For stability ($\Delta V_t \le 0$), the negative first-order term must dominate the positive second-order term. This requires $|e_{r,t} \Delta e_{r,t}| \ge \frac{1}{2}(\Delta e_{r,t})^2$, which simplifies to:
\begin{equation} \label{eq:stability_cond}
    2|e_{r,t}| \ge |\Delta e_{r,t}|
\end{equation}
Substituting Eq.~\eqref{eq:update_r} into Eq.~\eqref{eq:stability_cond} yields an upper bound on $\eta$. If this bound is satisfied, $\Delta V_t \le 0$, and by Lyapunov's direct method, $e_{r,t}$ is guaranteed to converge to zero.
\end{proof}

The stability condition derived in the proof leads to the following corollary regarding the learning rate.

\begin{cor}[Learning rate condition]
The convergence guaranteed by Lemma \ref{lem:convergence} is conditional on the learning rate $\eta$ satisfying the following bound at each step $t$:
\begin{equation} \label{eq:eta_bound}
\eta \le \frac{2 \cdot \hat{\theta}_{0,t} \cdot |e_{r,t}|^{1-\tau}}{p(r_t) \cdot (r_t-l_t)}
\end{equation}
where $p(r_t)$ is the true data density at the boundary point $r_t$.
\end{cor}

\begin{lem}[Convergence rate] \label{lem:rate}
    Let $0 < |e_{r,t}| < 1$. The proposed method achieves linear convergence when $\tau=1$. When $0 < \tau < 1$, the method achieves a superlinear convergence rate. When $\tau > 1$, the method achieves a sublinear convergence rate near the target.
\end{lem}
\begin{proof}
    The rate of convergence is determined by $V_t$. Substituting Eq. \eqref{eq:pr} and Eq. \eqref{eq:update_r} into Eq. \eqref{eq:delta_e}, the error change can be represented as:
    \begin{equation} \label{}
        \Delta e_{r,t} = - \eta \cdot \text{sign}(e_{r,t}) \cdot |e_{r,t}|^{\tau}
    \end{equation}
    Then, we have:
    \begin{equation}
        \Delta V_t = -\eta \cdot |e_{r,t}|^{1+\tau} + \frac{1}{2}\eta^2 \cdot |e_{r,t}|^{2\tau} 
    \end{equation}
    As the negative $\Delta V_t$ derives convergence, we analyze this expression for the cases of $\tau$.

    \textbf{i)} When $\tau=1$,
    \begin{align}
        \Delta V_t &= -\eta |e_{r,t}|^2 + \frac{1}{2} \eta^2 |e_{r,t}|^2 \\
        &= -\eta \left( 1 - \frac{1}{2} \eta \right) |e_{r,t}|^2 \\
        &= -\eta(2 - \eta) V_t
    \end{align}
    For sufficiently small $\eta$ such that $0 < \eta < 2$, $\Delta V_t < 0$, and $V_{t+1} = (1 - \eta)^2 V_t$ with $(1 - \eta)^2 < 1$. The error decreases by a constant factor at each step, defining linear convergence.

    \textbf{ii)} When $0 < \tau < 1$, the condition for convergence, $\Delta V_t < 0$, requires the negative term to dominate:
    $$\eta |e_{r,t}|^{1+\tau} > \frac{1}{2}\eta^2 |e_{r,t}|^{2\tau}$$
    Substituting $|e_{r,t}| = \sqrt{2V_t}$, we get:
    $$\eta (2V_t)^{\frac{1+\tau}{2}} > \frac{1}{2}\eta^2 (2V_t)^{\tau}$$
    $$\frac{2}{\eta} > (2V_t)^{\tau - \frac{1+\tau}{2}} = (2V_t)^{\frac{\tau-1}{2}}$$
    Since the exponent $\frac{\tau-1}{2}$ is negative, raising both sides to the power of $\frac{2}{\tau-1}$ reverses the inequality sign:
    $$\left(\frac{2}{\eta}\right)^{\frac{2}{\tau-1}} < 2V_t \implies V_t > \frac{1}{2}\left(\frac{\eta}{2}\right)^{\frac{2}{1-\tau}}$$
    Let $V^* = \frac{1}{2}(\frac{\eta}{2})^{\frac{2}{1-\tau}}$. For $V_t > V^*$ and small $\eta$, we have $\Delta V_t < - \frac{1}{2} \eta (2V_t)^{\frac{1+\tau}{2}}$. This recurrence implies superlinear convergence, as the relative decrement is larger for smaller $V_t$.
    
    \textbf{iii)} When $\tau > 1$:
    The convergence condition $\Delta V_t < 0$ again leads to the inequality:
    $$\frac{2}{\eta} > (2V_t)^{\frac{\tau-1}{2}}$$
    In this case, the exponent $\frac{\tau-1}{2}$ is positive, so the inequality sign is preserved when raising both sides to the power of $\frac{2}{\tau-1}$:
    $$\left(\frac{2}{\eta}\right)^{\frac{2}{\tau-1}} > 2V_t \implies V_t < \frac{1}{2}\left(\frac{2}{\eta}\right)^{\frac{2}{\tau-1}}$$
    Let $V^* = \frac{1}{2}(\frac{2}{\eta})^{\frac{2}{\tau-1}}$. For $V_t < V^*$ and small $\eta$, we have $\Delta V_t < - \frac{1}{2} \eta (2V_t)^{\frac{1+\tau}{2}}$. This recurrence implies sublinear convergence near zero, as the decrement becomes smaller for small $V_t$, dampening small errors.
\end{proof}

\begin{table*}[t!]
    \centering
    \caption{RMSE comparison across four datasets.}
    \label{tab:main_rmse}
    \scriptsize  
    \setlength{\tabcolsep}{2.2pt} 
    \renewcommand{\arraystretch}{1.1} 

    \begin{tabular}{ll ccccc c ccccc c ccccc c ccccc}
        \toprule
        & & \multicolumn{5}{c}{\textbf{Truncated Normal} ($\times 10^{-2}$)} & & \multicolumn{5}{c}{\textbf{HPC Voltage} ($\times 10^{-1}$)} & & \multicolumn{5}{c}{\textbf{Adult}} & & \multicolumn{5}{c}{\textbf{Employee} ($\times 10^4$)} \\
        \cmidrule{3-7} \cmidrule{9-13} \cmidrule{15-19} \cmidrule{21-25}
        Mech & $\epsilon$ & Base & ABC & TOPL & DPLAC & DCSGD & & Base & ABC & TOPL & DPLAC & DCSGD & & Base & ABC & TOPL & DPLAC & DCSGD & & Base & ABC & TOPL & DPLAC & DCSGD \\
        \midrule
        
        \multirow{5}{*}{Duchi}
          & 0.5 & 3.17 & 1.55 & 4.73 & \textbf{1.41} & 3.13 & & 2.84 & \textbf{1.30} & 5.20 & 1.89 & 3.69 & & 4.82 & \textbf{2.25} & 6.01 & 2.50 & 4.14 & & 5.17 & \textbf{0.49} & 20.26 & 0.61 & 5.32 \\
          & 1.0 & 2.82 & \textbf{0.73} & 10.3 & \textbf{0.73} & 1.45 & & 2.47 & \textbf{0.57} & 12.1 & 0.94 & 2.13 & & 3.85 & \textbf{0.83} & 11.7 & 1.20 & 1.49 & & 4.96 & \textbf{0.31} & 23.04 & 0.47 & 1.38 \\
          & 2.0 & 2.69 & \textbf{0.66} & 7.06 & 0.75 & 0.71 & & 2.32 & \textbf{0.26} & 5.14 & 0.45 & 1.60 & & 3.44 & \textbf{0.41} & 9.40 & 0.72 & 0.91 & & 4.86 & \textbf{0.22} & 25.87 & 0.36 & 0.56 \\
          & 3.0 & 2.66 & 0.64 & 6.17 & 0.70 & \textbf{0.61} & & 2.28 & \textbf{0.16} & 5.02 & 0.32 & 2.46 & & 3.35 & \textbf{0.29} & 8.80 & 0.56 & 1.14 & & 4.83 & \textbf{0.22} & 25.87 & 0.35 & 0.49 \\
          & 4.0 & 2.65 & 0.61 & 6.09 & 0.62 & \textbf{0.43} & & 2.28 & \textbf{0.16} & 5.01 & 0.28 & 1.86 & & 3.32 & \textbf{0.29} & 10.2 & 0.47 & 1.13 & & 4.82 & \textbf{0.21} & 25.87 & 0.35 & 0.52 \\
        \midrule
        
        \multirow{5}{*}{Three-output}
          & 0.5 & 3.24 & \textbf{1.45} & 4.70 & 1.81 & 2.94 & & 3.07 & \textbf{1.24} & 5.24 & 1.30 & 3.69 & & 4.96 & 2.89 & 5.66 & \textbf{1.95} & 4.47 & & 5.12 & \textbf{0.49} & 20.26 & 0.75 & 5.24 \\
          & 1.0 & 2.84 & \textbf{0.73} & 10.3 & 1.07 & 1.45 & & 2.49 & \textbf{0.56} & 12.1 & 0.88 & 2.02 & & 3.91 & \textbf{1.06} & 11.7 & 1.30 & 1.48 & & 4.93 & \textbf{0.31} & 23.04 & 0.52 & 1.42 \\
          & 2.0 & 2.71 & \textbf{0.61} & 7.06 & 0.71 & 0.71 & & 2.23 & \textbf{0.19} & 5.10 & 0.29 & 1.65 & & 3.23 & \textbf{0.42} & 9.38 & 0.66 & 1.01 & & 4.78 & \textbf{0.21} & 25.87 & 0.34 & 0.54 \\
          & 3.0 & 2.63 & \textbf{0.55} & 6.17 & 0.66 & 0.58 & & 2.20 & \textbf{0.15} & 4.98 & 0.22 & 2.44 & & 3.11 & \textbf{0.25} & 8.75 & 0.32 & 1.09 & & 4.74 & \textbf{0.20} & 25.87 & 0.28 & 0.50 \\
          & 4.0 & 2.61 & 0.59 & 6.09 & 0.64 & \textbf{0.40} & & 2.16 & \textbf{0.12} & 4.98 & 0.16 & 1.85 & & 3.07 & \textbf{0.20} & 10.2 & 0.28 & 1.08 & & 4.73 & \textbf{0.18} & 25.87 & 0.25 & 0.53 \\
        \midrule
        
        \multirow{5}{*}{PM}
          & 0.5 & 3.10 & \textbf{1.12} & 4.81 & 1.47 & 3.16 & & 2.91 & \textbf{0.84} & 5.15 & 1.95 & 3.64 & & 4.93 & \textbf{2.79} & 6.03 & 3.42 & 5.37 & & 5.02 & \textbf{0.50} & 20.22 & 0.78 & 5.27 \\
          & 1.0 & 2.81 & 0.74 & 10.3 & \textbf{0.66} & 1.54 & & 2.49 & \textbf{0.57} & 12.1 & 0.88 & 2.10 & & 3.81 & \textbf{1.31} & 11.7 & 1.38 & 1.65 & & 4.90 & \textbf{0.29} & 23.03 & 0.38 & 1.81 \\
          & 2.0 & 2.63 & \textbf{0.62} & 7.26 & 0.65 & 0.70 & & 2.27 & \textbf{0.23} & 5.13 & 0.43 & 1.64 & & 3.19 & \textbf{0.39} & 8.87 & 0.47 & 0.98 & & 4.77 & \textbf{0.24} & 25.87 & 0.35 & 0.57 \\
          & 3.0 & 2.59 & 0.62 & 6.15 & 0.61 & \textbf{0.59} & & 2.18 & \textbf{0.17} & 5.00 & 0.20 & 2.46 & & 3.07 & \textbf{0.29} & 8.14 & 0.36 & 1.07 & & 4.76 & \textbf{0.19} & 25.87 & 0.31 & 0.51 \\
          & 4.0 & 2.59 & 0.60 & 6.01 & 0.59 & \textbf{0.40} & & 2.16 & \textbf{0.12} & 4.97 & 0.20 & 1.84 & & 2.98 & 0.24 & 10.2 & \textbf{0.23} & 1.05 & & 4.73 & \textbf{0.20} & 25.87 & 0.25 & 0.53 \\
        \midrule
        
        \multirow{5}{*}{PM-SUB}
          & 0.5 & 3.07 & \textbf{0.88} & 4.80 & 1.16 & 3.21 & & 2.73 & \textbf{0.63} & 5.18 & 2.22 & 3.62 & & 5.02 & 2.95 & 5.87 & \textbf{2.64} & 4.32 & & 4.97 & \textbf{0.42} & 20.24 & 0.74 & 5.18 \\
          & 1.0 & 2.78 & \textbf{0.60} & 10.3 & 0.98 & 1.52 & & 2.81 & \textbf{0.56} & 12.1 & 0.74 & 1.98 & & 3.74 & \textbf{0.78} & 11.7 & 1.15 & 1.60 & & 4.91 & \textbf{0.36} & 23.03 & 0.43 & 1.81 \\
          & 2.0 & 2.63 & \textbf{0.59} & 7.34 & 0.60 & 0.71 & & 2.29 & \textbf{0.18} & 5.11 & 0.38 & 1.61 & & 3.17 & \textbf{0.36} & 9.41 & 0.45 & 0.93 & & 4.78 & \textbf{0.22} & 25.87 & 0.36 & 0.55 \\
          & 3.0 & 2.59 & 0.63 & 6.17 & 0.64 & \textbf{0.58} & & 2.20 & \textbf{0.17} & 4.99 & 0.28 & 2.42 & & 3.04 & \textbf{0.26} & 8.97 & 0.29 & 1.08 & & 4.74 & \textbf{0.20} & 25.87 & 0.26 & 0.50 \\
          & 4.0 & 2.60 & 0.60 & 6.08 & 0.60 & \textbf{0.41} & & 2.16 & \textbf{0.13} & 4.99 & 0.19 & 1.85 & & 2.99 & \textbf{0.20} & 9.58 & 0.30 & 1.12 & & 4.73 & \textbf{0.19} & 25.87 & 0.26 & 0.53 \\
        \midrule
        
        \multirow{5}{*}{HM-TP}
          & 0.5 & 3.29 & \textbf{1.03} & 4.65 & 1.51 & 3.59 & & 3.43 & \textbf{0.82} & 5.20 & 1.36 & 3.41 & & 4.70 & 2.45 & 5.84 & \textbf{1.87} & 4.72 & & 5.08 & \textbf{0.44} & 20.25 & 0.81 & 5.35 \\
          & 1.0 & 2.75 & \textbf{0.66} & 10.2 & 1.10 & 1.41 & & 2.39 & \textbf{0.64} & 12.0 & 0.81 & 2.09 & & 3.68 & \textbf{0.95} & 11.7 & 1.09 & 1.56 & & 4.90 & \textbf{0.30} & 23.04 & 0.53 & 1.49 \\
          & 2.0 & 2.69 & \textbf{0.66} & 7.06 & 0.68 & 0.74 & & 2.28 & \textbf{0.23} & 5.14 & 0.40 & 1.59 & & 3.18 & \textbf{0.42} & 9.40 & 0.45 & 0.89 & & 4.78 & \textbf{0.22} & 25.87 & 0.33 & 0.55 \\
          & 3.0 & 2.60 & 0.61 & 6.17 & 0.59 & \textbf{0.60} & & 2.17 & \textbf{0.16} & 4.98 & 0.24 & 2.45 & & 3.08 & \textbf{0.24} & 8.06 & 0.30 & 1.12 & & 4.75 & \textbf{0.19} & 25.87 & 0.28 & 0.50 \\
          & 4.0 & 2.56 & 0.57 & 6.15 & 0.57 & \textbf{0.41} & & 2.17 & \textbf{0.13} & 4.98 & 0.23 & 1.84 & & 3.01 & \textbf{0.19} & 9.69 & 0.22 & 1.07 & & 4.73 & \textbf{0.19} & 25.87 & 0.25 & 0.53 \\
        \midrule
        
        \multirow{5}{*}{N-output}
          & 0.5 & 2.99 & \textbf{1.08} & 4.73 & 1.16 & 3.06 & & 2.90 & \textbf{0.67} & 5.16 & 1.06 & 3.15 & & 3.82 & 4.23 & 5.69 & \textbf{2.01} & 4.06 & & 5.06 & \textbf{0.54} & 20.25 & 0.70 & 5.19 \\
          & 1.0 & 2.74 & \textbf{0.74} & 10.2 & 0.82 & 1.41 & & 2.64 & \textbf{0.40} & 12.1 & 0.64 & 2.14 & & 3.48 & 0.90 & 11.7 & \textbf{0.89} & 1.61 & & 4.89 & \textbf{0.32} & 23.04 & 0.47 & 1.97 \\
          & 2.0 & 2.61 & \textbf{0.64} & 6.93 & 0.67 & 0.69 & & 2.32 & \textbf{0.24} & 5.12 & 0.45 & 1.67 & & 3.17 & \textbf{0.34} & 8.36 & 0.48 & 0.97 & & 4.80 & \textbf{0.24} & 25.87 & 0.35 & 0.55 \\
          & 3.0 & 2.59 & 0.61 & 6.17 & 0.60 & \textbf{0.56} & & 2.19 & \textbf{0.16} & 4.99 & 0.26 & 2.46 & & 3.05 & \textbf{0.22} & 8.06 & 0.34 & 1.11 & & 4.76 & \textbf{0.21} & 25.87 & 0.29 & 0.50 \\
          & 4.0 & 2.60 & 0.60 & 6.08 & 0.60 & \textbf{0.40} & & 2.15 & \textbf{0.13} & 4.98 & 0.18 & 1.85 & & 2.99 & \textbf{0.19} & 9.21 & 0.25 & 1.08 & & 4.73 & \textbf{0.19} & 25.87 & 0.27 & 0.53 \\
        \midrule
        
        \multirow{5}{*}{AAA}
          & 0.5 & 2.80 & 0.88 & 4.71 & \textbf{0.84} & 2.75 & & 2.58 & \textbf{0.28} & 5.07 & 0.47 & 2.61 & & 2.70 & 1.74 & 4.15 & \textbf{0.98} & 2.65 & & 4.87 & \textbf{0.34} & 20.2 & 0.41 & 4.71 \\
          & 1.0 & 2.68 & 0.80 & 10.2 & \textbf{0.76} & 1.32 & & 2.35 & \textbf{0.30} & 12.0 & 0.38 & 2.01 & & 2.48 & 0.83 & 11.6 & \textbf{0.66} & 1.35 & & 4.81 & \textbf{0.22} & 23.0 & 0.32 & 2.15 \\
          & 2.0 & 2.60 & \textbf{0.65} & 6.94 & 0.67 & 0.86 & & 2.20 & \textbf{0.18} & 5.11 & 0.20 & 1.61 & & 2.32 & \textbf{0.36} & 8.31 & 0.40 & 0.80 & & 4.74 & \textbf{0.20} & 25.9 & 0.27 & 0.52 \\
          & 3.0 & 2.58 & \textbf{0.61} & 6.17 & \textbf{0.61} & 0.75 & & 2.15 & \textbf{0.12} & 5.00 & 0.13 & 2.42 & & 2.31 & \textbf{0.17} & 8.03 & 0.30 & 0.51 & & 4.72 & \textbf{0.18} & 25.8 & 0.25 & 0.46 \\
          & 4.0 & 2.57 & \textbf{0.60} & 6.08 & \textbf{0.60} & 0.70 & & 2.13 & \textbf{0.10} & 4.98 & 0.12 & 1.82 & & 2.29 & \textbf{0.15} & 9.19 & 0.25 & 0.31 & & 4.71 & \textbf{0.17} & 25.8 & 0.23 & 0.48 \\
        \bottomrule
    \end{tabular}
\end{table*}

\section{Experiments} \label{sec:exp}
In this section, we present the results of experiments conducted under various settings to validate the effectiveness of the ABC method. Our objective is to demonstrate that ABC method improves the performance of existing LDP mechanisms for numerical data collection, particularly in environments where no prior domain knowledge is available.


\subsection{Experimental Setup}
\textbf{Datasets.} Our evaluation uses one synthetic dataset, Truncated Normal, and three real-world datasets, Adult \cite{adult_2}, Employee \cite{SFEmployeeCompensation}, and HPC Voltage \cite{individual_household_electric_power_consumption_235}. Except for the main results, all subsequent experiments are reported on the Truncated Normal and HPC Voltage datasets using PM-SUB as the representative mechanism, due to space limitations.

\textbf{Baseline LDP mechanisms.} We employ six numerical LDP mechanisms as baselines: Duchi's mechanism (Duchi) \cite{duchi2018minimax}, the Three-output mechanism \cite{zhao2020local}, the Piecewise mechanism (PM) \cite{wang2019collecting}, a Suboptimal Piecewise mechanism (PM-SUB), the Hybrid mechanism Three-output and PM-SUB (HM-TP) \cite{zhao2020local}, the N-output mechanism \cite{baek2025noutput},  and AAA mechanism \cite{wei2024aaa}.

\textbf{Evaluation task.} 
We use mean estimation as the benchmark task to evaluate performance. More than just a simple statistical task, mean estimation here serves to quantify the overall error between the true statistics of the original data and the statistics recovered from the perturbed data. Initially, the server guesses an arbitrary domain without prior knowledge. For the baselines, the server estimates the mean from perturbed data collected using this fixed domain. In contrast, our ABC method estimates the mean while dynamically adjusting the domain. The evaluation metric is the Root Mean Square Error (RMSE) between the true mean and the estimated mean. A lower RMSE indicates a more accurate estimation and higher data utility. All experiments are repeated 10 times, and we report the average of these runs.

\textbf{Privacy budget.} We conduct experiments across a range of privacy budgets, $\epsilon \in \{0.5, 1, 2, 3, 4\}$, to analyze performance under different privacy constraints. 

\textbf{Hyperparameter settings.} 
Unless otherwise specified, we use the following default hyperparameter settings: target clipping ratio $\alpha=0.05$ to manage the bias-variance trade-off; learning rate $\eta=0.3$; number of rounds $T=30$; privacy budget ratio $\beta=0.7$ for splitting $\epsilon$; stability constant $\zeta=0.1$; and error amplification parameter $\tau=1/2$.


\subsection{Mean Estimation}
To evaluate the ABC method, we applied it to six baseline LDP mechanisms and measured the Root Mean Square Error (RMSE) on four datasets. To simulate practical scenarios where the data domain is unknown, we set the initial bounds using five different scale factors, $\mathrm{scale} \in \{\frac{1}{8}, \frac{1}{4}, \frac{1}{2}, \frac{2}{3}, 1, 1.5, 2, 4, 8\}$. Let the true center of the data be $c^* = (l^* + r^*) / 2$, where $l^*$ and $r^*$ are the true minimum and maximum values, respectively. The initial bounds $[l_0, r_0]$ are determined by scaling the distance from this center point:
\begin{equation}
l_0 = c^* - \mathrm{scale} \cdot (c^* - l^*), \quad r_0 = c^* + \mathrm{scale} \cdot (r^* - c^*)
\end{equation}
Unless otherwise specified, all reported results are averaged over these five initial scale settings. 

Table \ref{tab:main_rmse} shows the RMSE comparison results. The ABC method consistently achieves significantly lower RMSE than the Base and TOPL baselines across all datasets and mechanisms, demonstrating the robustness of our adaptive domain estimation strategy in unknown domain settings. Compared to state-of-the-art methods like DPLAC and DCSGD, ABC shows highly competitive performance. While these baseline methods occasionally yield competitive results in mean estimation, they exhibit significantly degraded performance in distribution estimation, as demonstrated in Section \ref{sec:distribution}. In contrast, ABC consistently maintains high accuracy across both metrics, validating its superiority as a more robust solution.
\begin{table*}[t!]
    \centering
    \caption{Wasserstein Distance comparison. (Lower is better)}
    \label{tab:ws_metric}
    \footnotesize 
    \setlength{\tabcolsep}{3pt} 
    \renewcommand{\arraystretch}{1.1} 

    \begin{tabular}{ll ccccc c ccccc c ccccc c ccccc}
        \toprule
        & & \multicolumn{5}{c}{\textbf{Truncated Normal} ($\times 10^{-2}$)} & & \multicolumn{5}{c}{\textbf{HPC Voltage}} & & \multicolumn{5}{c}{\textbf{Adult}} & & \multicolumn{5}{c}{\textbf{Employee} ($\times 10^4$)} \\
        \cmidrule{3-7} \cmidrule{9-13} \cmidrule{15-19} \cmidrule{21-25}
        Mech & $\epsilon$ & Base & ABC & TOPL & DPL & DCS & & Base & ABC & TOPL & DPL & DCS & & Base & ABC & TOPL & DPL & DCS & & Base & ABC & TOPL & DPL & DCS \\
        \midrule
        
        \multirow{5}{*}{PM}
          & 0.5 & 4.16 & \textbf{3.80} & 13.2 & 20.6 & 24.5 & & 0.83 & \textbf{0.50} & 2.57 & 9.33 & 11.3 & & 6.42 & \textbf{6.10} & 16.1 & 20.4 & 25.9 & & 6.81 & \textbf{1.08} & 26.5 & 29.3 & 35.3 \\
          & 1.0 & 4.28 & \textbf{1.47} & 13.4 & 20.7 & 2.77 & & 1.15 & \textbf{0.36} & 2.77 & 9.34 & 1.03 & & 5.70 & \textbf{2.87} & 16.9 & 20.0 & 4.74 & & 8.23 & \textbf{0.95} & 26.0 & 30.6 & 2.22 \\
          & 2.0 & 5.66 & \textbf{0.99} & 13.2 & 19.7 & 4.49 & & 1.84 & \textbf{0.34} & 2.40 & 9.19 & 1.63 & & 6.90 & \textbf{2.02} & 16.3 & 19.6 & 6.41 & & 11.6 & \textbf{0.95} & 26.0 & 30.0 & 2.73 \\
          & 3.0 & 6.51 & \textbf{1.06} & 13.1 & 13.6 & 4.77 & & 2.27 & \textbf{0.38} & 2.40 & 6.16 & 1.74 & & 7.63 & \textbf{1.92} & 16.3 & 13.8 & 6.73 & & 12.7 & \textbf{1.00} & 26.0 & 20.0 & 3.39 \\
          & 4.0 & 5.51 & \textbf{1.23} & 13.1 & 8.95 & 4.03 & & 1.77 & \textbf{0.41} & 2.40 & 4.07 & 1.41 & & 7.04 & \textbf{1.92} & 16.5 & 9.39 & 5.59 & & 11.2 & \textbf{1.14} & 26.0 & 13.3 & 3.02 \\
        \midrule
        
        \multirow{5}{*}{PM-SUB}
          & 0.5 & 3.74 & \textbf{2.60} & 13.2 & 19.5 & 24.7 & & 0.66 & \textbf{0.51} & 2.57 & 9.29 & 11.0 & & 6.02 & \textbf{7.53} & 16.1 & 23.0 & 26.0 & & 5.97 & \textbf{1.20} & 26.6 & 32.2 & 35.3 \\
          & 1.0 & 3.22 & \textbf{1.55} & 13.4 & 21.0 & 2.82 & & 0.60 & \textbf{0.47} & 2.77 & 9.51 & 1.04 & & 5.10 & \textbf{3.02} & 16.9 & 22.0 & 4.57 & & 5.72 & \textbf{1.21} & 26.0 & 30.1 & 2.17 \\
          & 2.0 & 2.92 & \textbf{1.19} & 13.2 & 20.5 & 4.49 & & 0.55 & \textbf{0.44} & 2.40 & 9.37 & 1.62 & & 4.36 & \textbf{2.23} & 16.3 & 20.2 & 6.39 & & 5.61 & \textbf{0.99} & 26.0 & 30.4 & 2.74 \\
          & 3.0 & 2.86 & \textbf{1.24} & 13.1 & 16.4 & 4.77 & & 0.54 & \textbf{0.44} & 2.40 & 7.42 & 1.73 & & 4.15 & \textbf{1.84} & 16.3 & 16.6 & 6.67 & & 5.59 & \textbf{1.02} & 26.0 & 24.2 & 3.38 \\
          & 4.0 & 2.80 & \textbf{1.16} & 13.1 & 11.2 & 4.04 & & 0.53 & \textbf{0.42} & 2.40 & 5.15 & 1.41 & & 4.04 & \textbf{1.71} & 16.5 & 11.5 & 5.66 & & 5.55 & \textbf{0.98} & 26.0 & 16.9 & 3.03 \\
        \bottomrule
    \end{tabular}
\end{table*}

\begin{table*}[t!]
    \centering
    \caption{Kolmogorov-Smirnov (KS) Statistic comparison. (Lower is better)}
    \label{tab:ks_metric}
    \footnotesize 
    \setlength{\tabcolsep}{3pt} 
    \renewcommand{\arraystretch}{1.1} 

    \begin{tabular}{ll ccccc c ccccc c ccccc c ccccc}
        \toprule
        & & \multicolumn{5}{c}{\textbf{Truncated Normal}} & & \multicolumn{5}{c}{\textbf{HPC Voltage}} & & \multicolumn{5}{c}{\textbf{Adult}} & & \multicolumn{5}{c}{\textbf{Employee}} \\
        \cmidrule{3-7} \cmidrule{9-13} \cmidrule{15-19} \cmidrule{21-25}
        Mech & $\epsilon$ & Base & ABC & TOPL & DPL & DCS & & Base & ABC & TOPL & DPL & DCS & & Base & ABC & TOPL & DPL & DCS & & Base & ABC & TOPL & DPL & DCS \\
        \midrule
        
        \multirow{5}{*}{PM}
          & 0.5 & 0.29 & \textbf{0.11} & 0.57 & 0.24 & 0.41 & & 0.17 & \textbf{0.08} & 0.39 & 0.25 & 0.36 & & 0.33 & \textbf{0.14} & 0.53 & 0.21 & 0.41 & & 0.43 & \textbf{0.09} & 0.90 & 0.33 & 0.54 \\
          & 1.0 & 0.30 & \textbf{0.08} & 0.79 & 0.22 & 0.19 & & 0.19 & \textbf{0.06} & 0.55 & 0.24 & 0.17 & & 0.31 & \textbf{0.11} & 0.74 & 0.21 & 0.25 & & 0.47 & \textbf{0.08} & 0.94 & 0.33 & 0.24 \\
          & 2.0 & 0.30 & \textbf{0.06} & 0.67 & 0.20 & 0.25 & & 0.20 & \textbf{0.05} & 0.39 & 0.23 & 0.24 & & 0.32 & \textbf{0.09} & 0.64 & 0.20 & 0.23 & & 0.49 & \textbf{0.09} & 1.00 & 0.32 & 0.19 \\
          & 3.0 & 0.31 & \textbf{0.06} & 0.63 & 0.17 & 0.26 & & 0.20 & \textbf{0.05} & 0.38 & 0.21 & 0.26 & & 0.33 & \textbf{0.09} & 0.62 & 0.19 & 0.26 & & 0.50 & \textbf{0.10} & 1.00 & 0.29 & 0.23 \\
          & 4.0 & 0.30 & \textbf{0.07} & 0.62 & 0.14 & 0.19 & & 0.19 & \textbf{0.06} & 0.38 & 0.16 & 0.19 & & 0.36 & \textbf{0.09} & 0.68 & 0.16 & 0.22 & & 0.49 & \textbf{0.10} & 1.00 & 0.25 & 0.21 \\
        \midrule
        
        \multirow{5}{*}{PM-SUB}
          & 0.5 & 0.27 & \textbf{0.10} & 0.57 & 0.22 & 0.41 & & 0.13 & \textbf{0.07} & 0.39 & 0.25 & 0.36 & & 0.33 & \textbf{0.17} & 0.53 & 0.23 & 0.41 & & 0.38 & \textbf{0.10} & 0.90 & 0.33 & 0.54 \\
          & 1.0 & 0.26 & \textbf{0.10} & 0.79 & 0.22 & 0.19 & & 0.13 & \textbf{0.07} & 0.55 & 0.24 & 0.18 & & 0.30 & \textbf{0.13} & 0.74 & 0.22 & 0.23 & & 0.38 & \textbf{0.13} & 0.94 & 0.33 & 0.23 \\
          & 2.0 & 0.25 & \textbf{0.06} & 0.68 & 0.21 & 0.25 & & 0.12 & \textbf{0.06} & 0.39 & 0.24 & 0.24 & & 0.28 & \textbf{0.11} & 0.64 & 0.21 & 0.23 & & 0.38 & \textbf{0.10} & 1.00 & 0.32 & 0.19 \\
          & 3.0 & 0.24 & \textbf{0.06} & 0.63 & 0.18 & 0.25 & & 0.12 & \textbf{0.06} & 0.38 & 0.22 & 0.26 & & 0.28 & \textbf{0.09} & 0.64 & 0.19 & 0.25 & & 0.38 & \textbf{0.11} & 1.00 & 0.30 & 0.23 \\
          & 4.0 & 0.24 & \textbf{0.06} & 0.62 & 0.15 & 0.19 & & 0.12 & \textbf{0.06} & 0.38 & 0.18 & 0.19 & & 0.27 & \textbf{0.08} & 0.66 & 0.17 & 0.22 & & 0.38 & \textbf{0.10} & 1.00 & 0.27 & 0.21 \\
        \bottomrule
    \end{tabular}
\end{table*}

\begin{figure*}[!ht]
\centering

\begin{center}
    \includegraphics[width=0.6\textwidth]{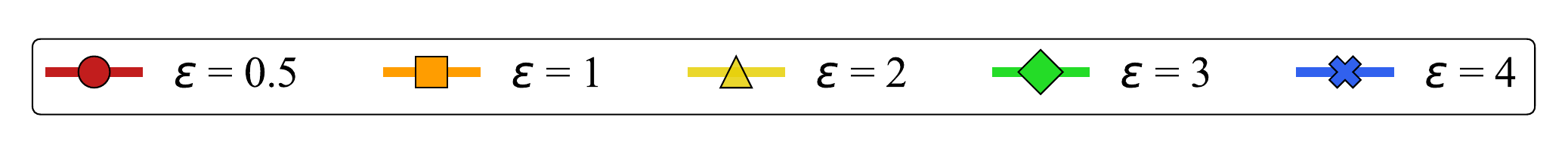}
\end{center}
\vspace{-0.2cm}

\begin{subfigure}{0.5\columnwidth}
\centering
\includegraphics[width=\linewidth, trim=0 0 0 0, clip]{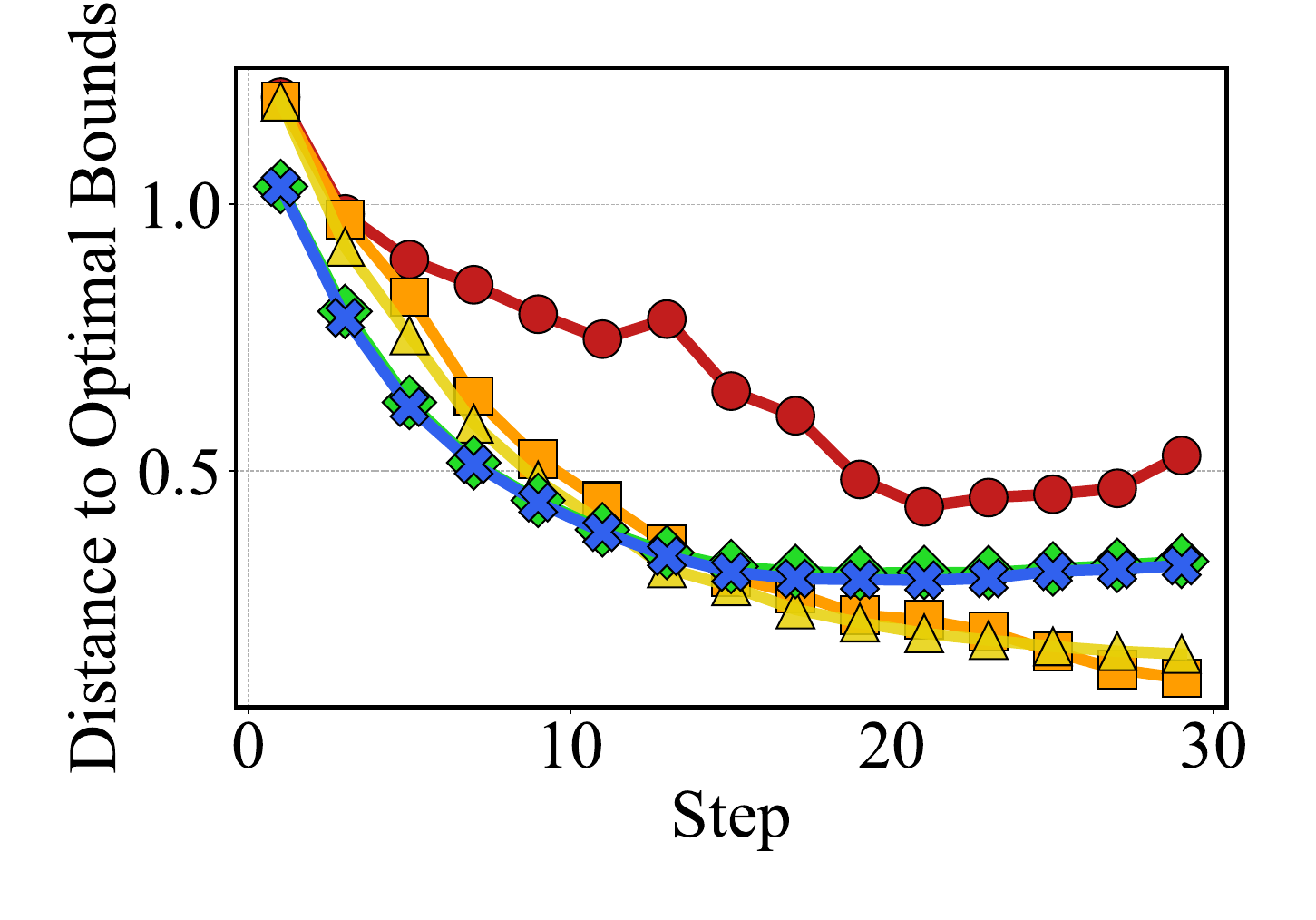}
\caption{Truncated Normal}
\end{subfigure}
\begin{subfigure}{0.5\columnwidth}
\centering
\includegraphics[width=\linewidth, trim=0 0 0 0, clip]{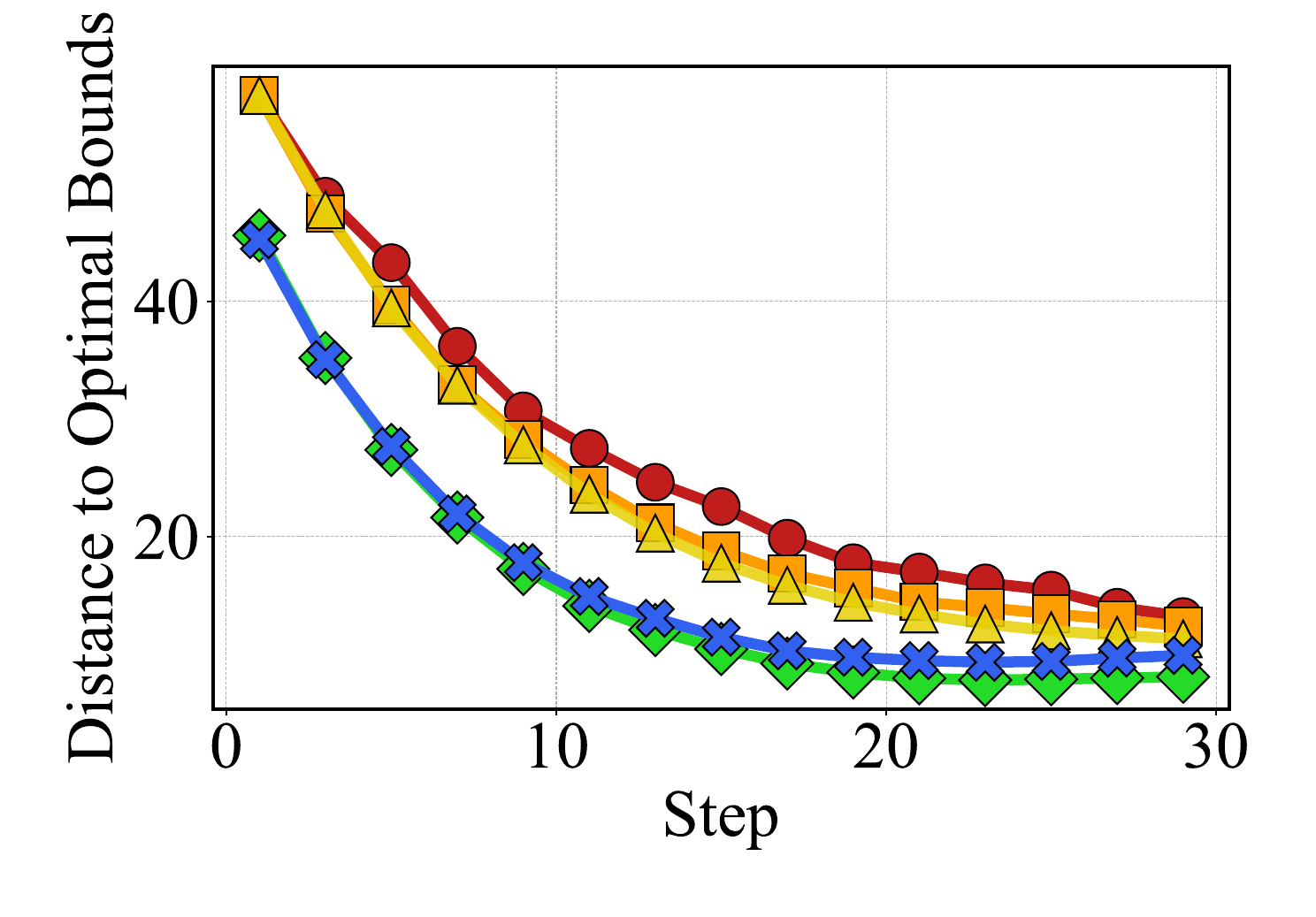}
\caption{HPC Voltage}
\end{subfigure}
\begin{subfigure}{0.5\columnwidth}
\centering
\includegraphics[width=\linewidth, trim=0 0 0 0, clip]{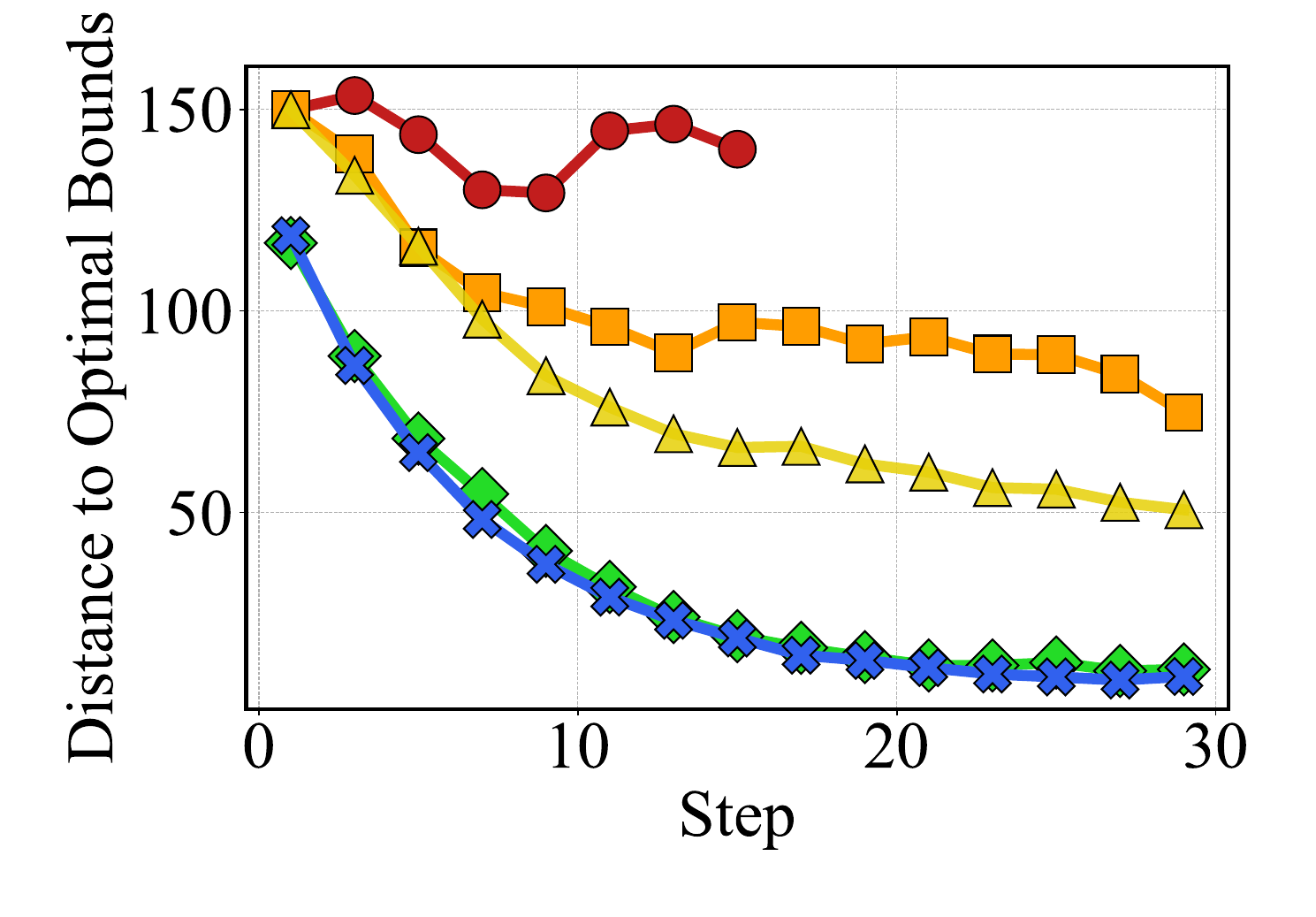}
\caption{Adult}
\end{subfigure}
\begin{subfigure}{0.5\columnwidth}
\centering
\includegraphics[width=\linewidth, trim=0 0 0 0, clip]{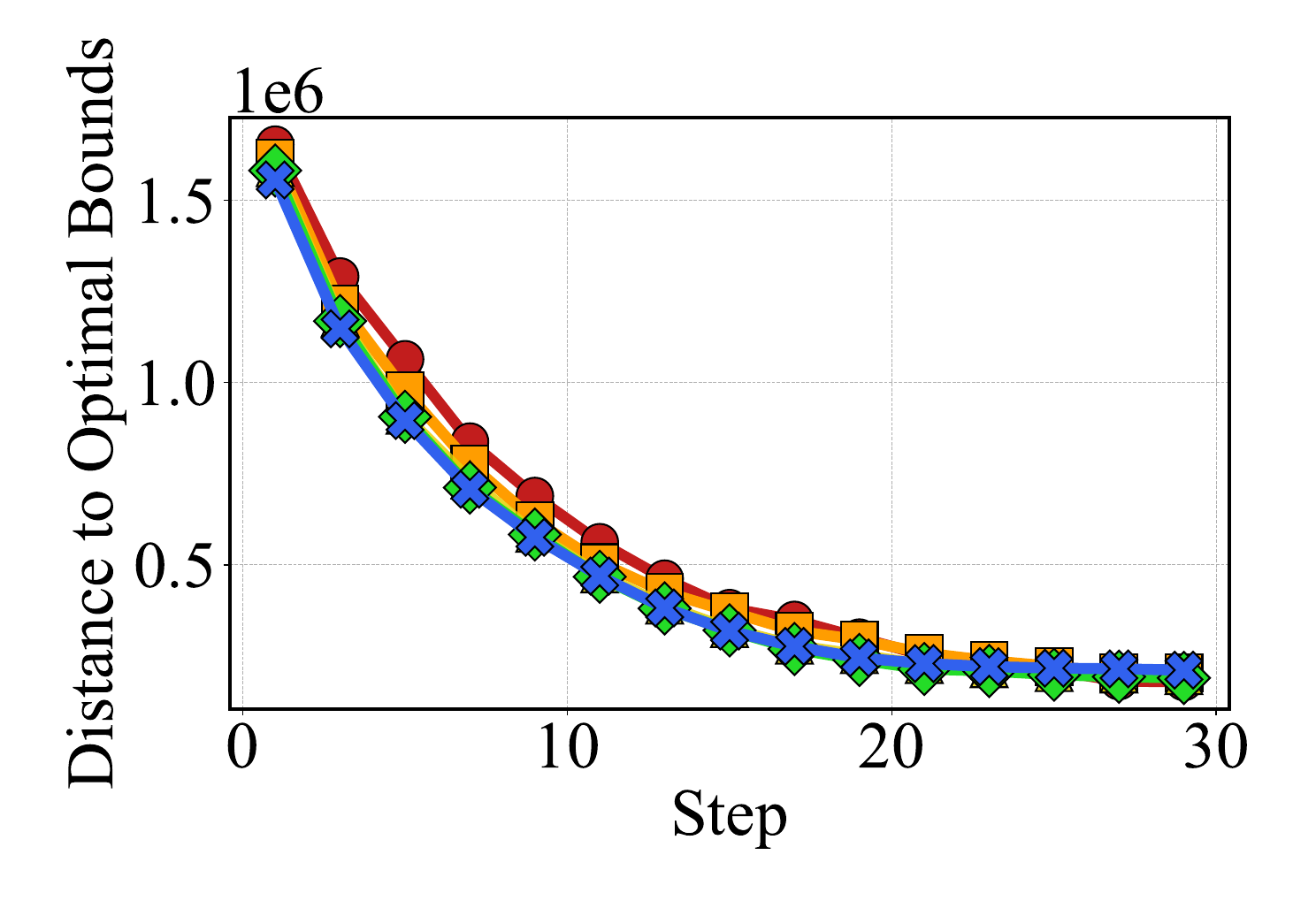}
\caption{Employee}
\end{subfigure}

\centering
\caption{Analysis of the optimality gap for the estimated bounds. We measure the distance to the optimal bounds ($|l - l_{opt}| + |r - r_{opt}|$) over optimization steps. The results demonstrate that our mechanism effectively converges to the optimal bounds, with faster convergence observed at larger privacy budgets ($\epsilon$).}
\label{fig:03-optimal-bounds}
\end{figure*}

\subsection{Distribution Estimation} \label{sec:distribution}
To evaluate how well the mechanisms preserve the underlying data distribution, we employed two complementary metrics: the Wasserstein Distance, which quantifies the minimum cost to transform the estimated distribution into the true one, and the Kolmogorov-Smirnov (KS) statistic, which measures the maximum distance between Cumulative Distribution Functions (CDFs).

Tables \ref{tab:ws_metric} and \ref{tab:ks_metric} present the comparison results. ABC consistently achieves the lowest errors in both Wasserstein and KS metrics across all datasets. Notably, in terms of Wasserstein distance, ABC frequently outperforms baselines by a significant margin, demonstrating its ability to accurately recover both the shape and location of the data mass.

The results further clarify the limitations of DPLAC and DCSGD. While these methods showed competitive performance in mean estimation, they exhibit severe degradation in distribution metrics. This gap is particularly pronounced in the Wasserstein distance (Table \ref{tab:ws_metric}), especially on the Employee dataset. This confirms that their optimization is narrowly focused on point estimates (mean), whereas ABC provides a holistic and robust estimation of the entire data distribution.

\subsection{Convergence to Optimal Bound}
To quantify the effectiveness of our adaptive strategy, we define the \textit{optimality gap} at step $t$. This metric measures the total deviation of the estimated bounds $[l_t, r_t]$ from the true optimal bounds $[l_{opt}, r_{opt}]$:
\begin{equation}
\mathrm{Gap}_t = |l_t - l_{opt}| + |r_t - r_{opt}|
\end{equation}
Figure \ref{fig:03-optimal-bounds} demonstrate a monotonic decrease in the gap across iterations for all settings. This confirms that the estimated bounds consistently converge toward the theoretical optimal bounds. The rate of convergence is positively correlated with $\epsilon$. With a larger $\epsilon$ (lower noise), the gap decreases steeply, allowing the algorithm to rapidly approach the optimal bounds. Notably, for the Adult dataset at $\epsilon=0.5$, the algorithm automatically calibrates the total number of steps based on the calculated noise variance. This adaptive mechanism ensures that the method remains robust even in high-noise regimes, preventing divergence.




\begin{figure*}[ht]
\centering
\begin{subfigure}{0.99\columnwidth}
    \centering
    \includegraphics[width=1.1\linewidth]{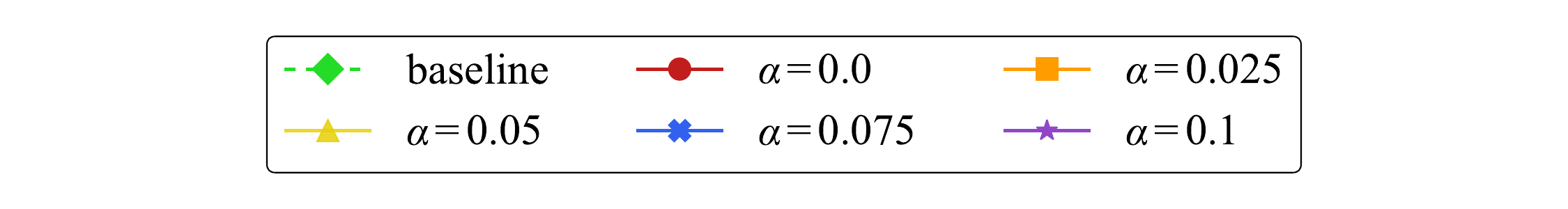}
    \begin{subfigure}{0.49\columnwidth}
        \centering
        \includegraphics[width=\linewidth, trim=0 0 0 0, clip]{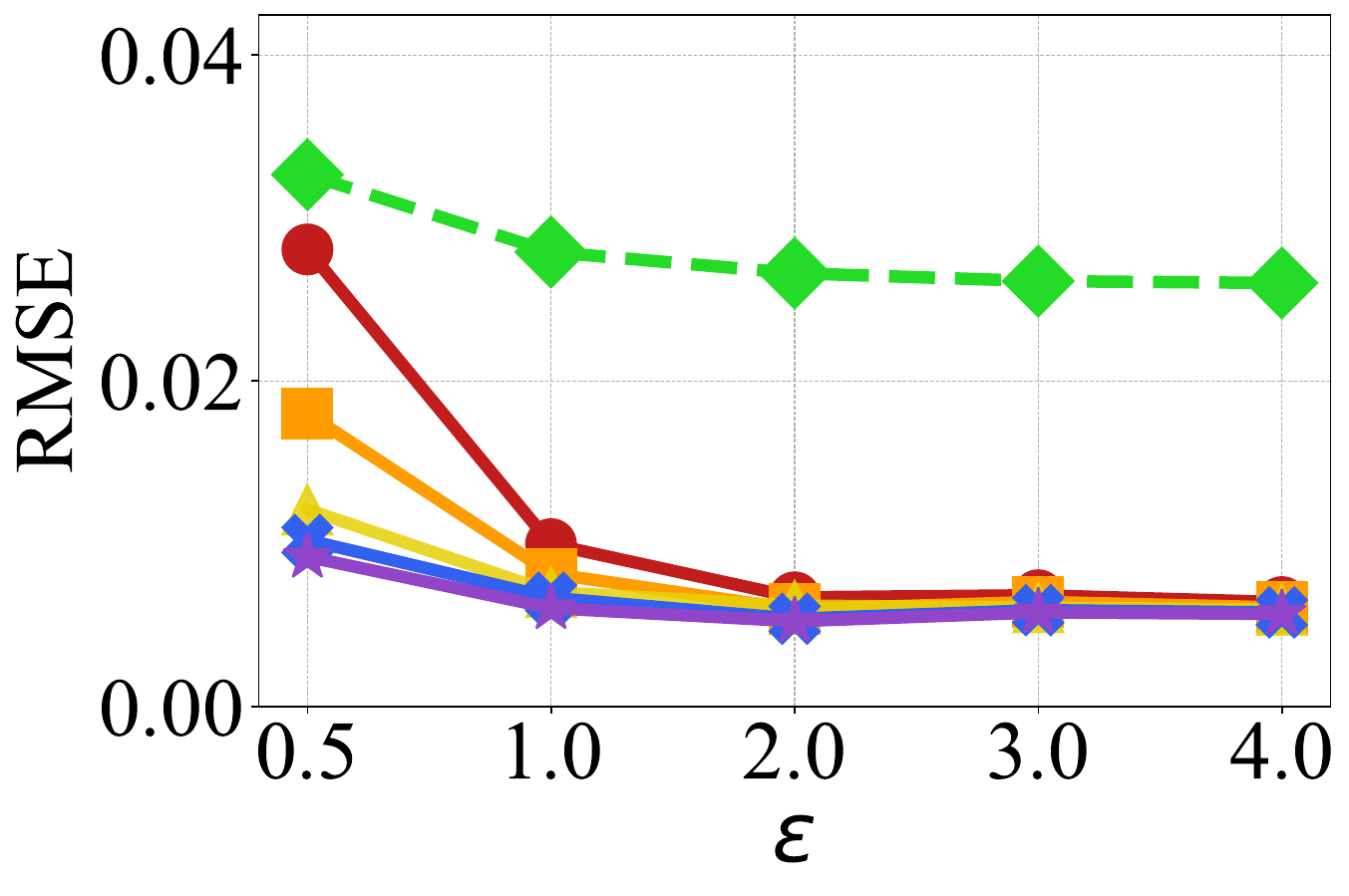}
        \caption{Truncated Normal ($\alpha$)}
        \label{fig:alpha-tn}
    \end{subfigure}
    \hfill
    \begin{subfigure}{0.49\columnwidth}
        \centering
        \includegraphics[width=\linewidth, trim=0 0 0 0, clip]{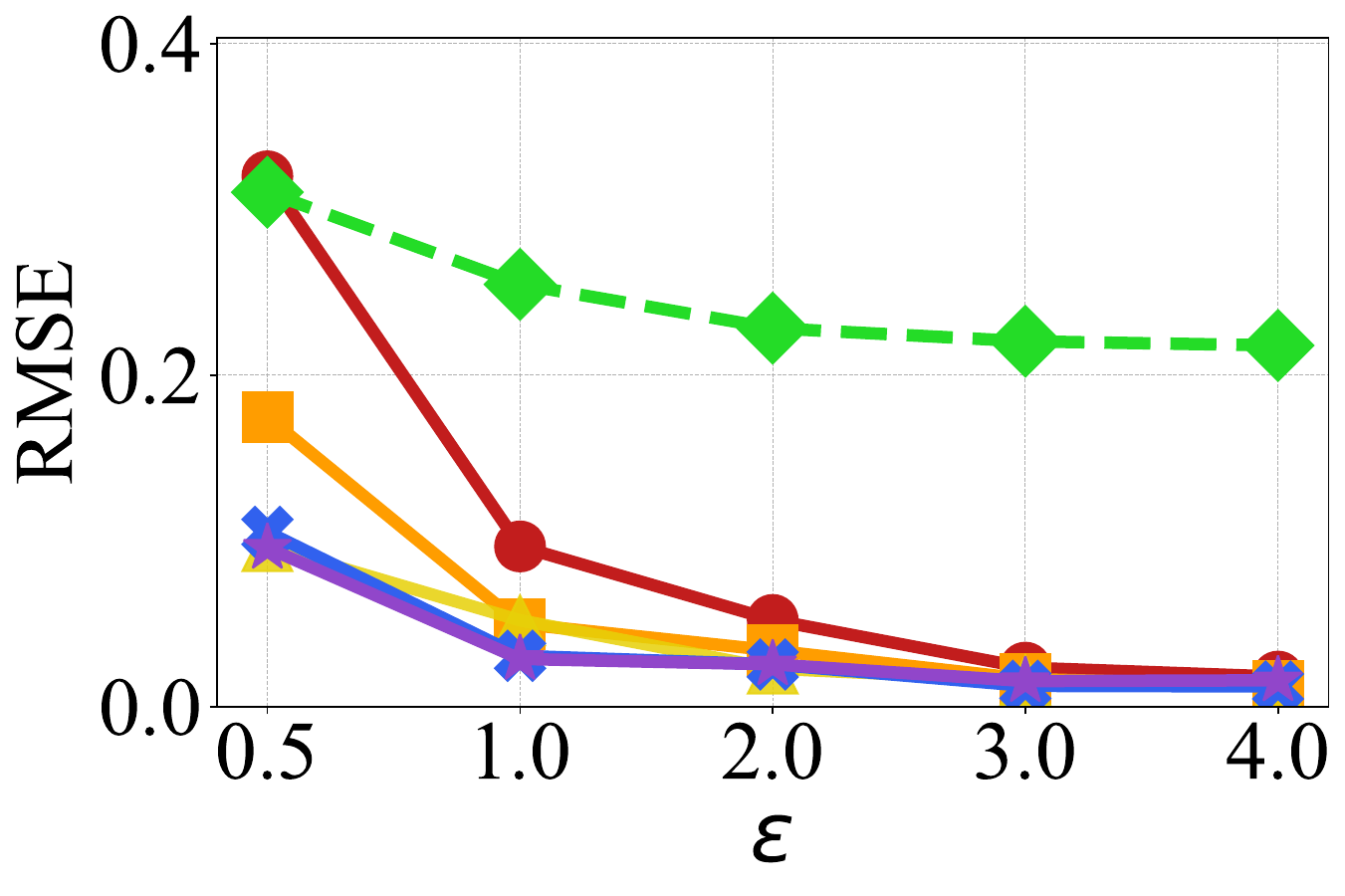}
        \caption{HPC Voltage ($\alpha$)}
        \label{fig:alpha-hpc}
    \end{subfigure}
\end{subfigure}
\hfill
\begin{subfigure}{0.99\columnwidth}
    \centering
    \includegraphics[width=1.1\linewidth]{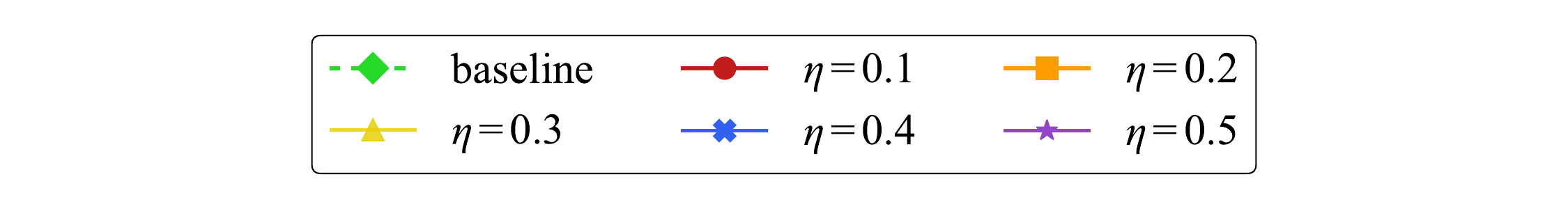}
    \begin{subfigure}{0.49\columnwidth}
        \centering
        \includegraphics[width=\linewidth, trim=0 0 0 0, clip]{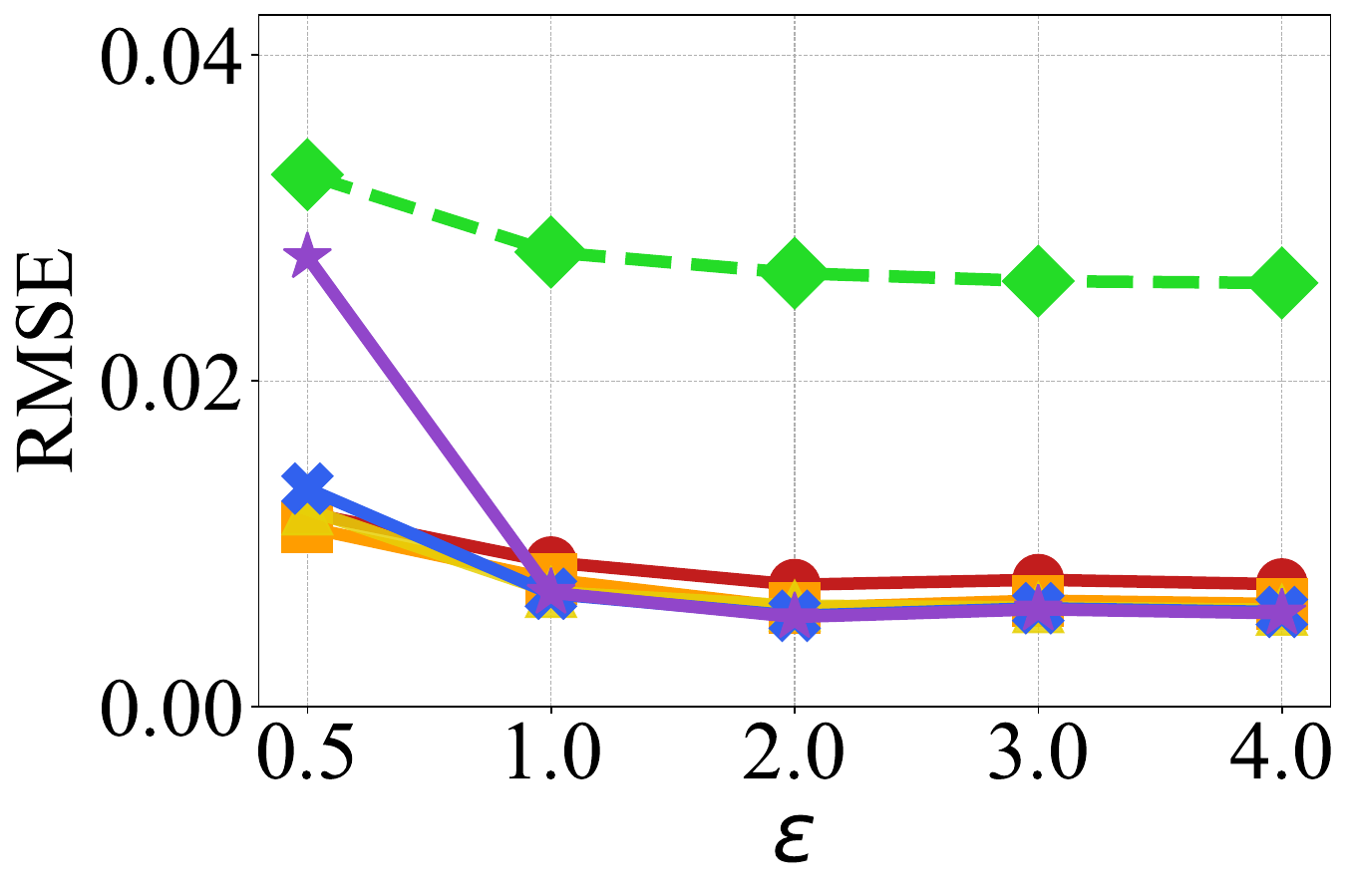}
        \caption{Truncated Normal ($\eta$)}
        \label{fig:lr-tn}
    \end{subfigure}
    \hfill
    \begin{subfigure}{0.49\columnwidth}
        \centering
        \includegraphics[width=\linewidth, trim=0 0 0 0, clip]{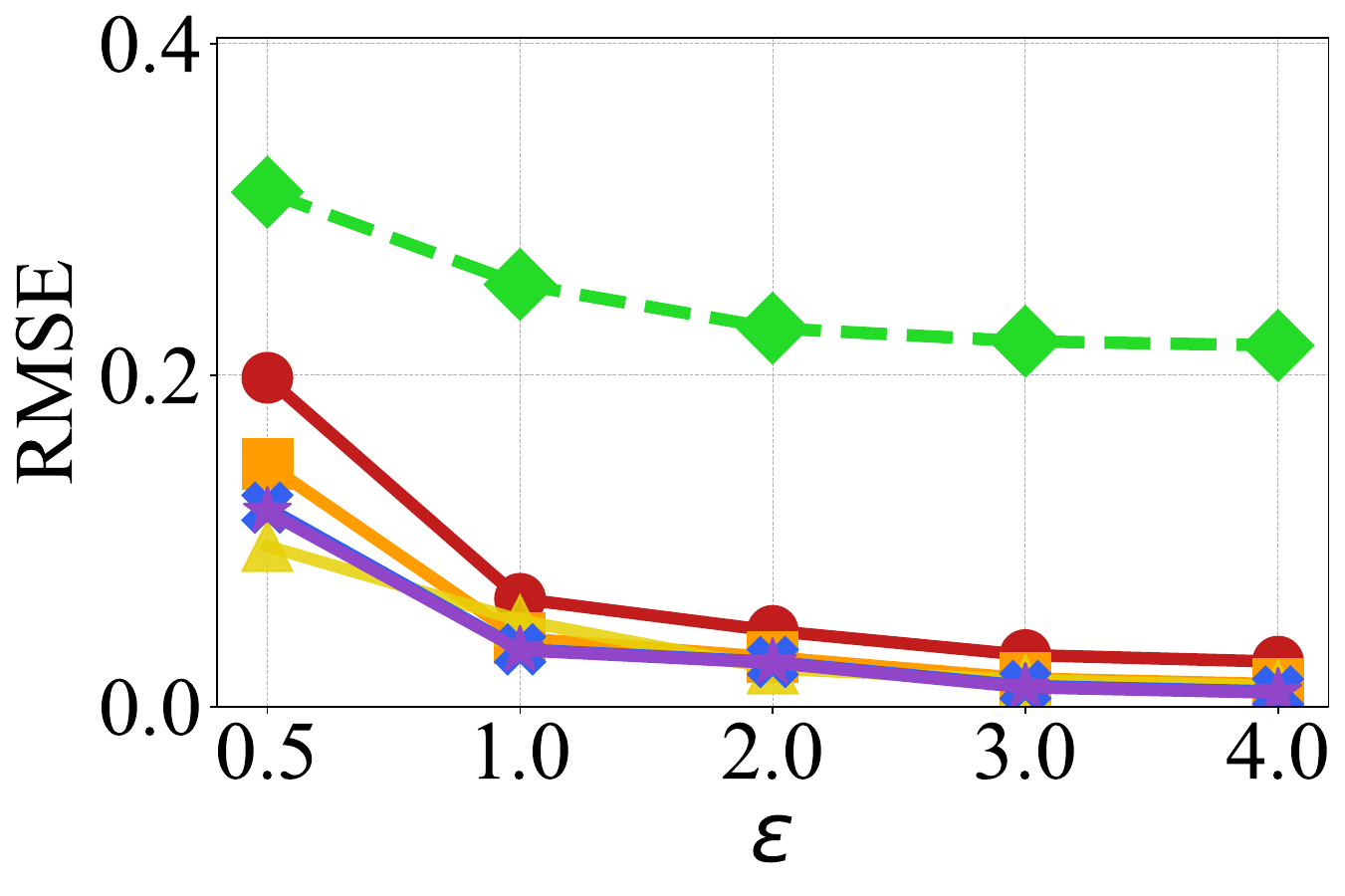}
        \caption{HPC Voltage ($\eta$)}
        \label{fig:lr-hpc}
    \end{subfigure}
\end{subfigure}

\begin{subfigure}{0.99\columnwidth}
    \centering
    \includegraphics[width=1.1\linewidth]{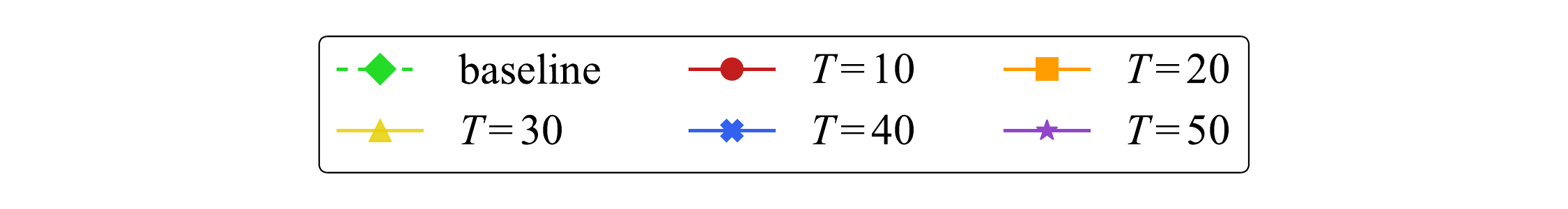}
    \begin{subfigure}{0.49\columnwidth}
        \centering
        \includegraphics[width=\linewidth, trim=0 0 0 0, clip]{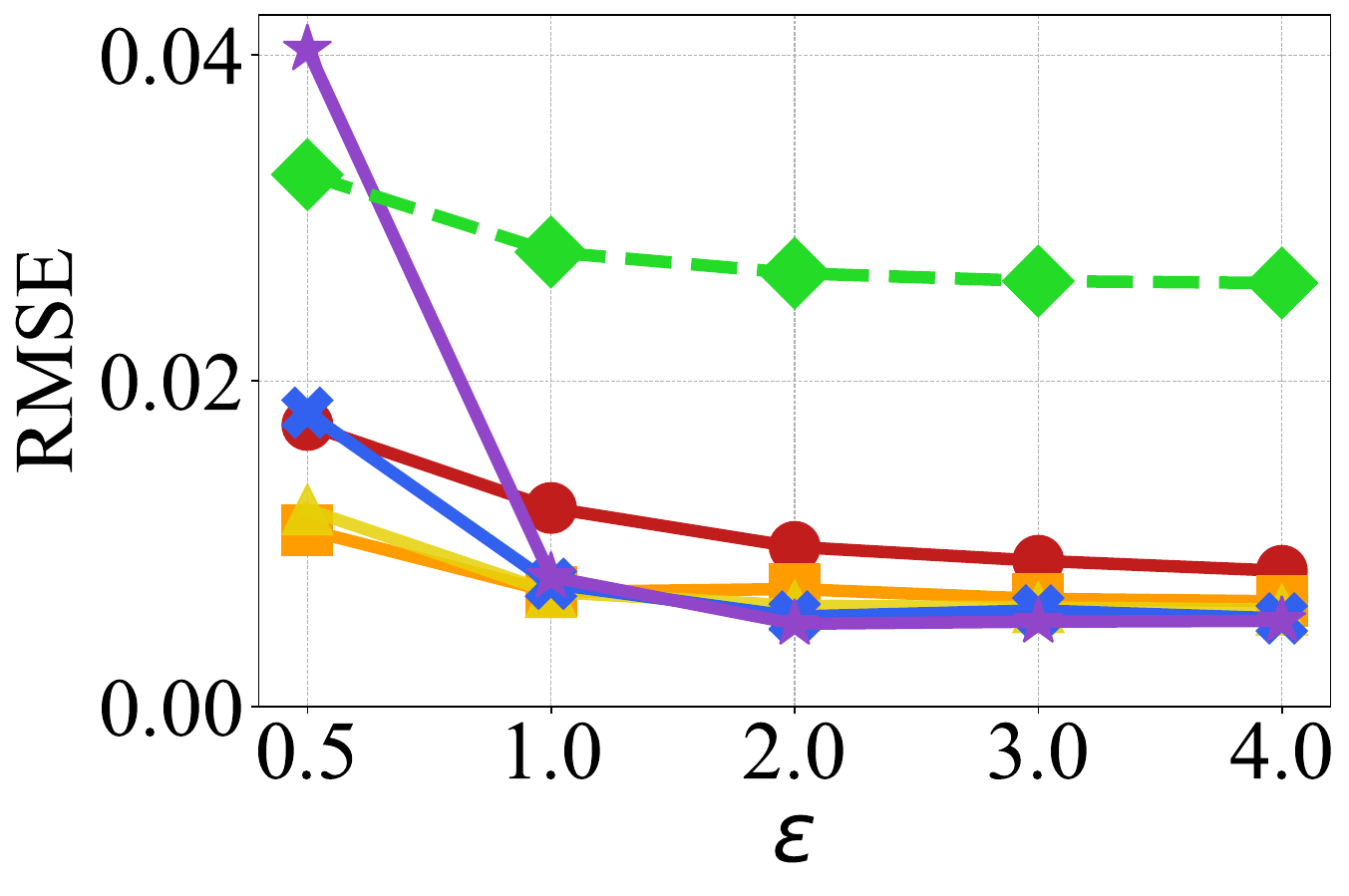}
        \caption{Truncated Normal ($T$)}
        \label{fig:step-tn}
    \end{subfigure}
    \hfill
    \begin{subfigure}{0.49\columnwidth}
        \centering
        \includegraphics[width=\linewidth, trim=0 0 0 0, clip]{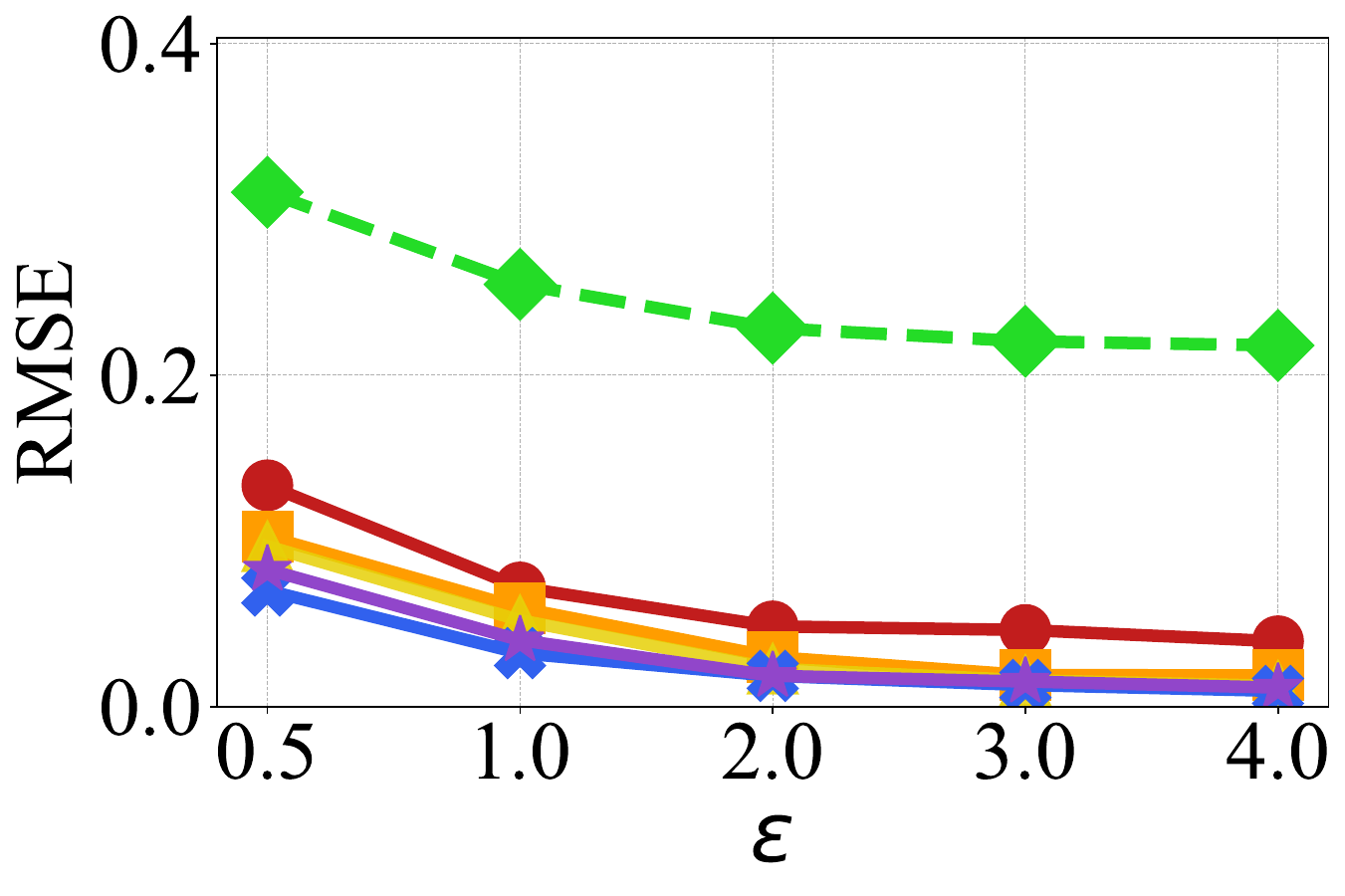}
        \caption{HPC Voltage ($T$)}
        \label{fig:step-hpc}
    \end{subfigure}
\end{subfigure}
\hfill
\begin{subfigure}{0.99\columnwidth}
    \centering
    \includegraphics[width=1.1\linewidth]{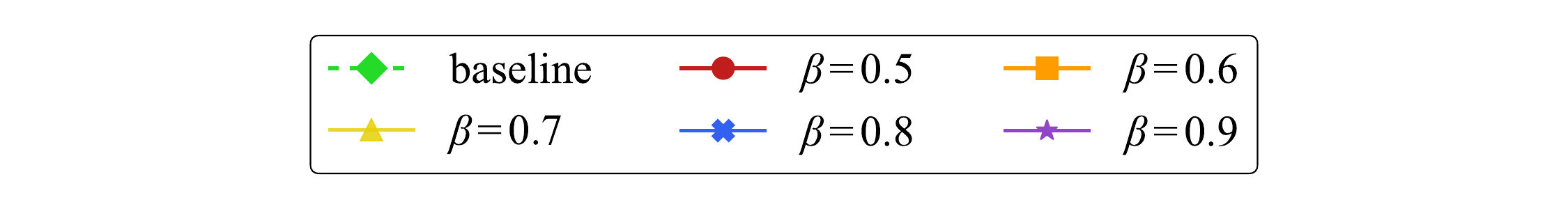}
    \begin{subfigure}{0.49\columnwidth}
        \centering
        \includegraphics[width=\linewidth, trim=0 0 0 0, clip]{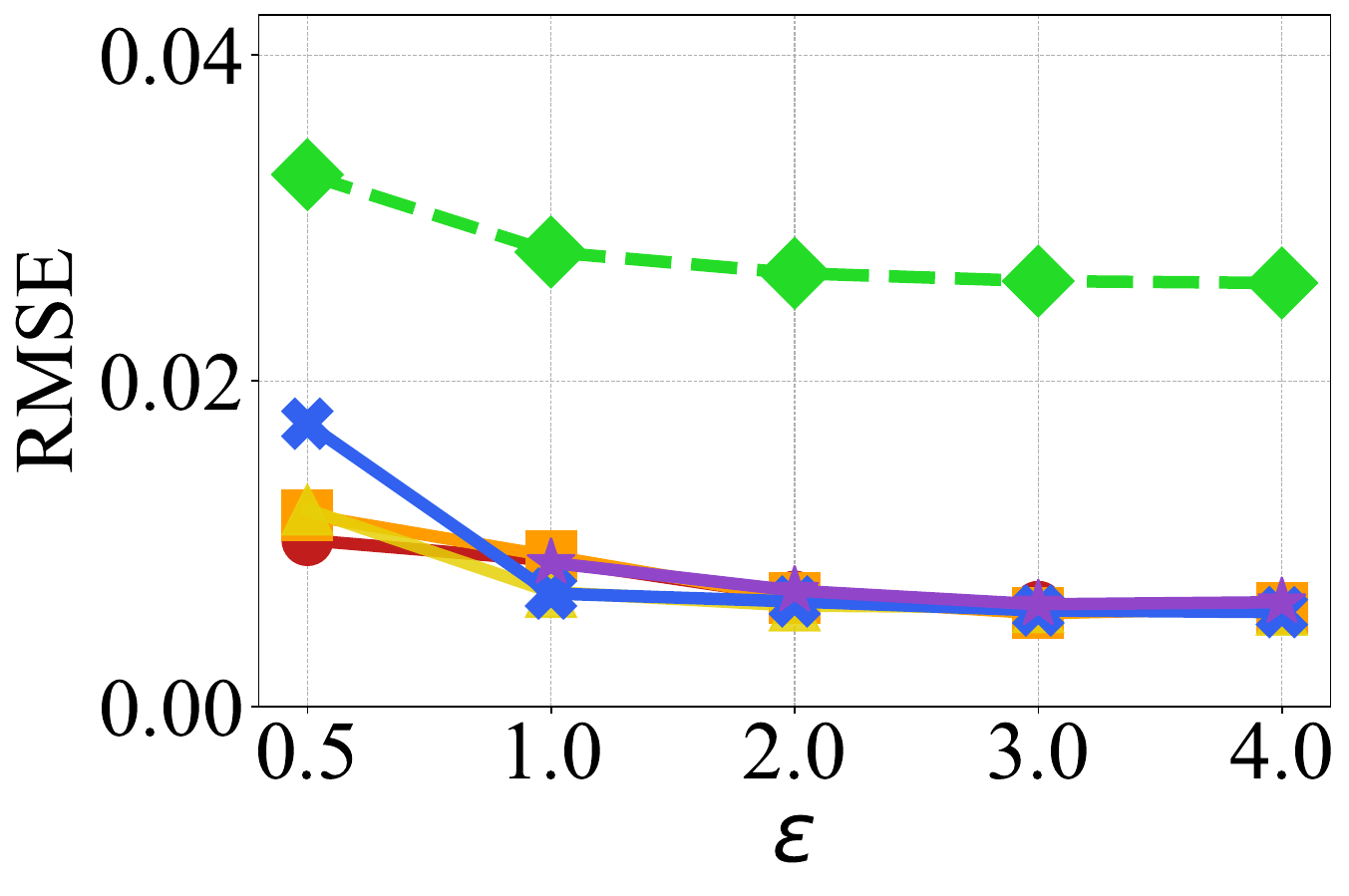}
        \caption{Truncated Normal ($\beta$)}
        \label{fig:beta-tn}
    \end{subfigure}
    \hfill
    \begin{subfigure}{0.49\columnwidth}
        \centering
        \includegraphics[width=\linewidth, trim=0 0 0 0, clip]{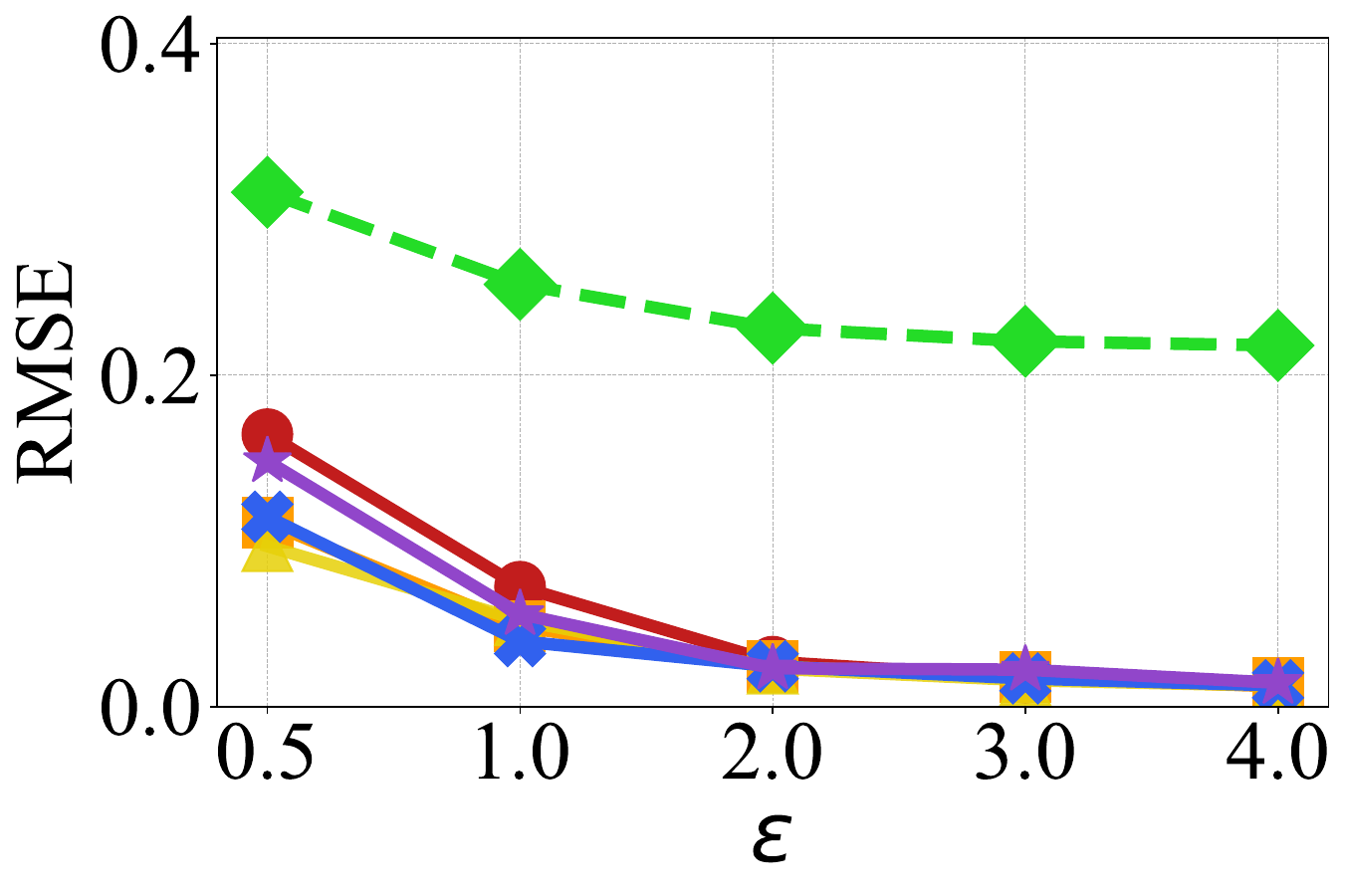}
        \caption{HPC Voltage ($\beta$)}
        \label{fig:beta-hpc}
    \end{subfigure}
\end{subfigure}

\begin{subfigure}{0.99\columnwidth}
    \centering
    \includegraphics[width=1.1\linewidth]{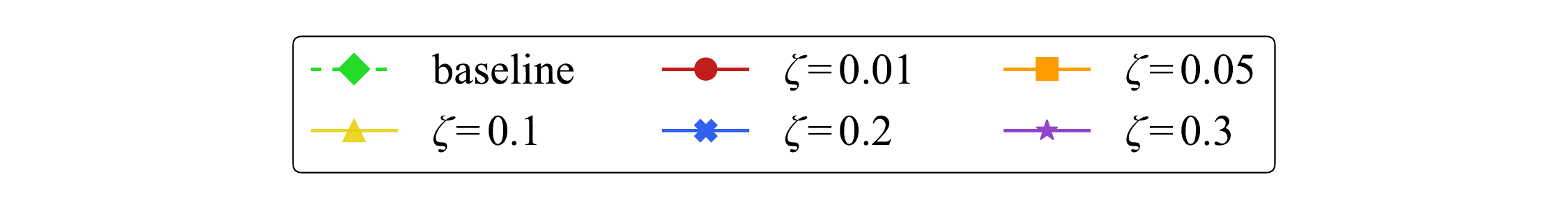}
    \begin{subfigure}{0.49\columnwidth}
        \centering
        \includegraphics[width=\linewidth, trim=0 0 0 0, clip]{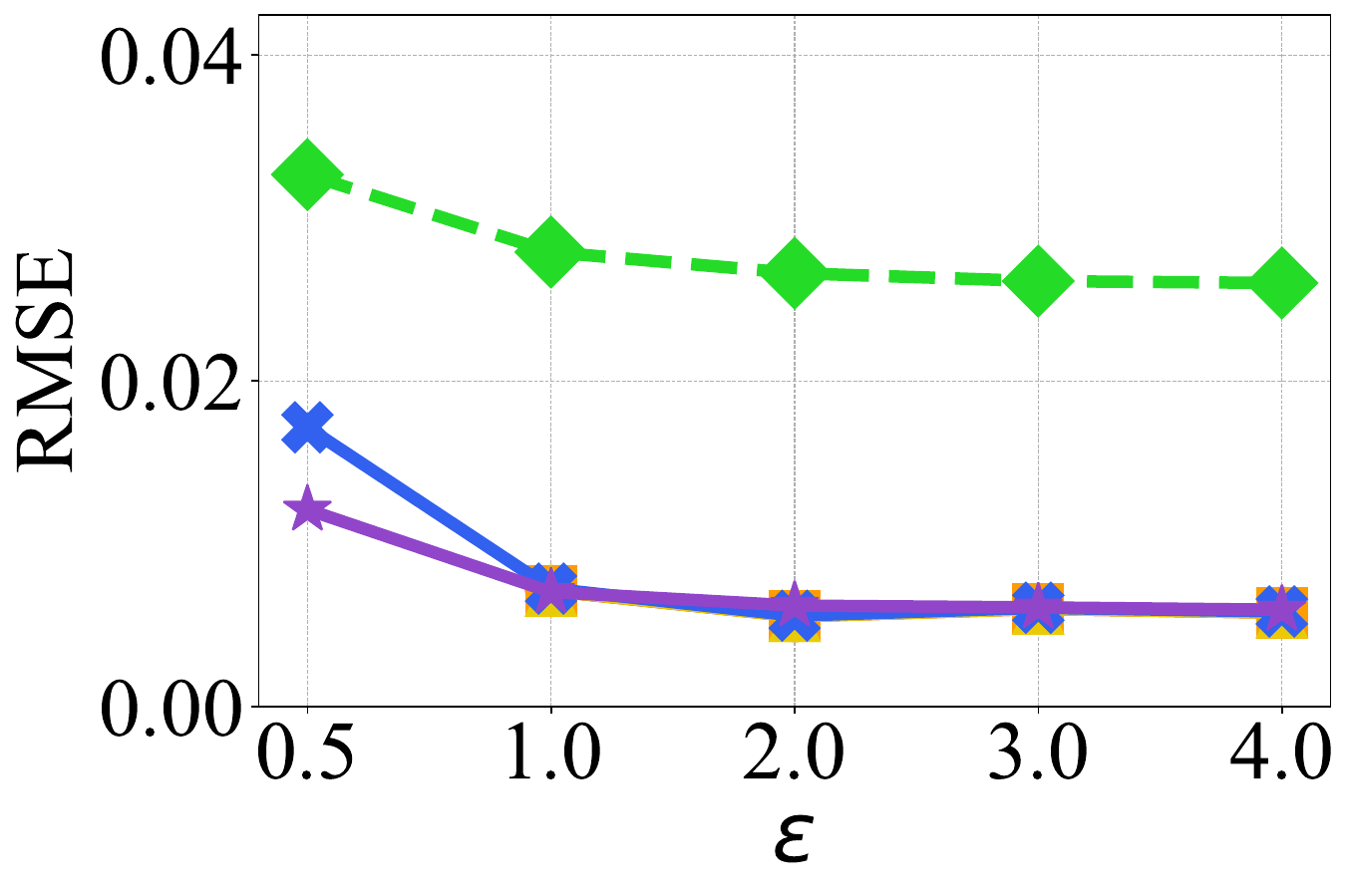}
        \caption{Truncated Normal ($\zeta$)}
        \label{fig:zeta-tn}
    \end{subfigure}
    \hfill
    \begin{subfigure}{0.49\columnwidth}
        \centering
        \includegraphics[width=\linewidth, trim=0 0 0 0, clip]{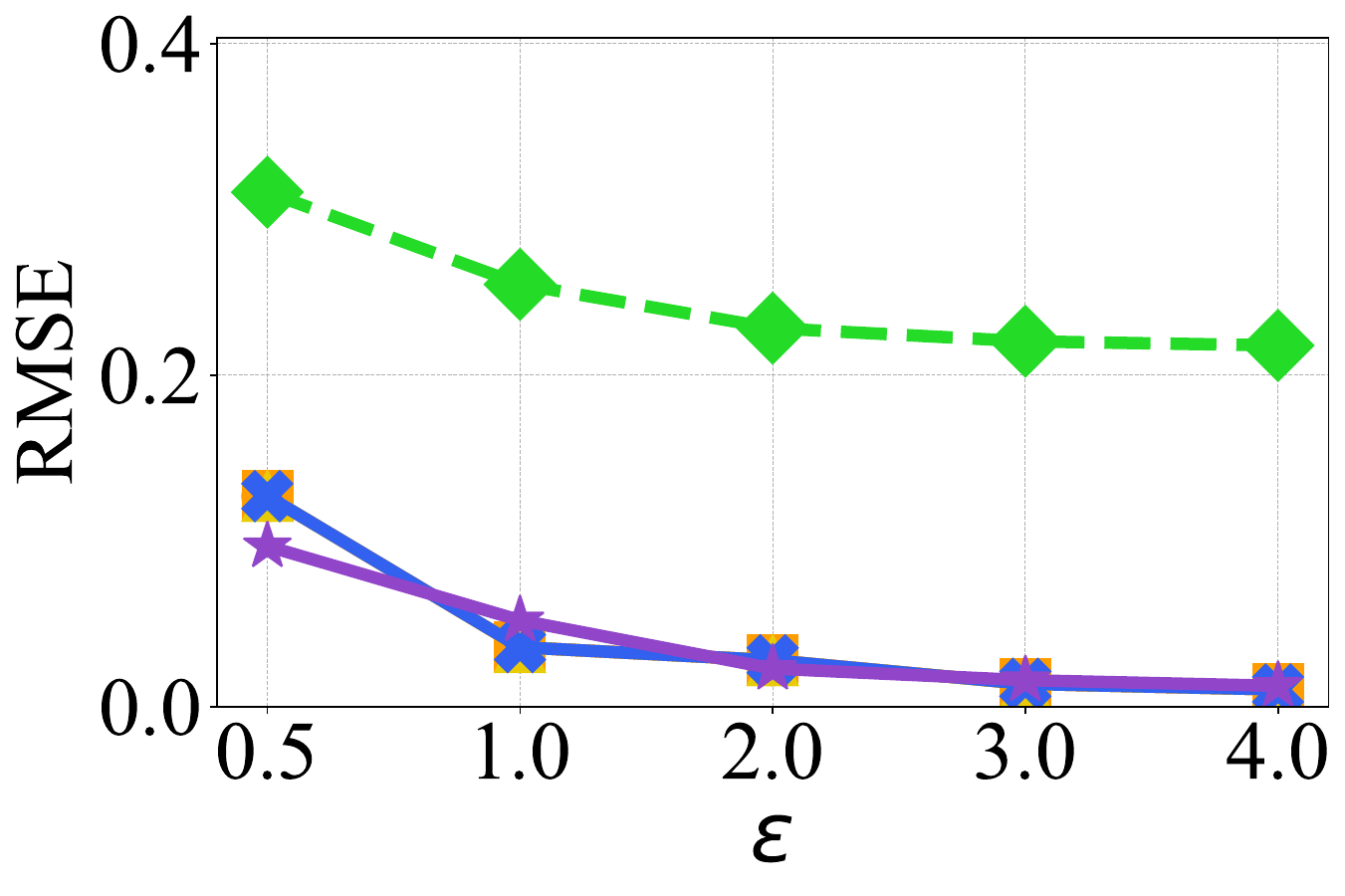}
        \caption{HPC Voltage ($\zeta$)}
        \label{fig:zeta-hpc}
    \end{subfigure}
\end{subfigure}
\hfill
\begin{subfigure}{0.99\columnwidth}
    \centering
    \includegraphics[width=1.1\linewidth]{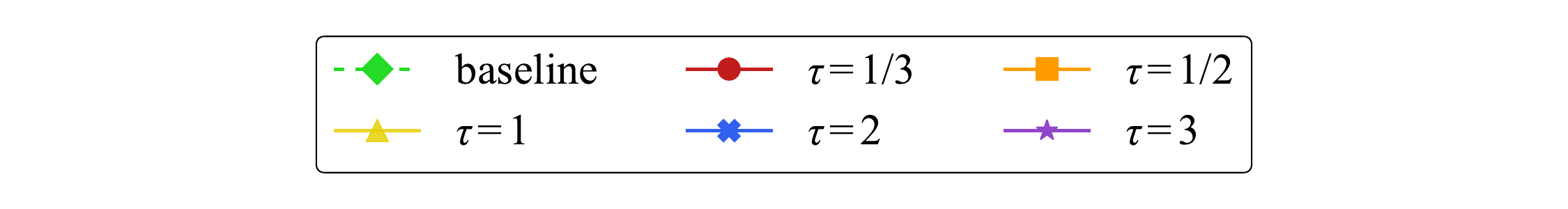}
    \begin{subfigure}{0.49\columnwidth}
        \centering
        \includegraphics[width=\linewidth, trim=0 0 0 0, clip]{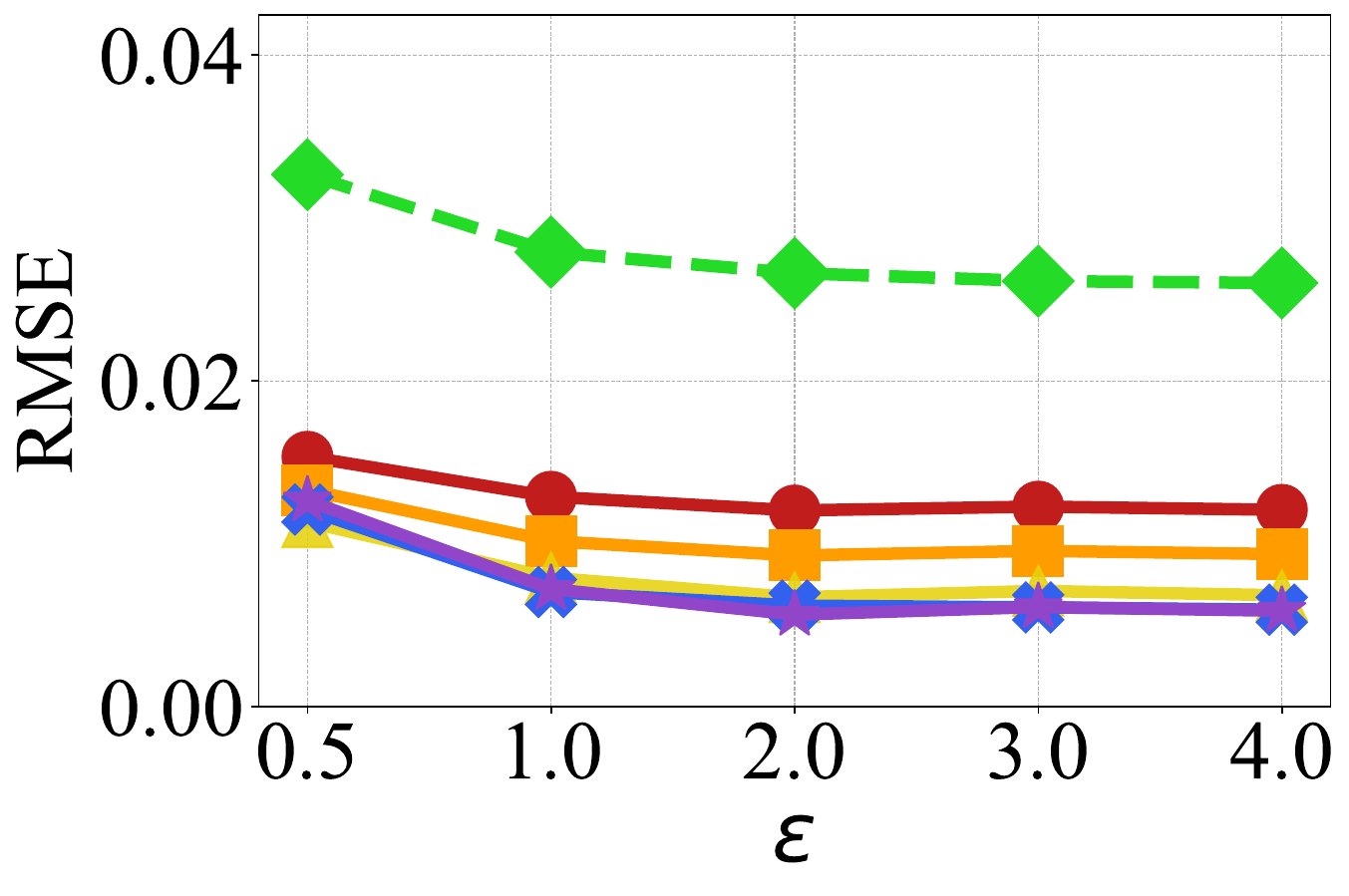}
        \caption{Truncated Normal ($\tau$)}
        \label{fig:tau-tn}
    \end{subfigure}
    \hfill
    \begin{subfigure}{0.49\columnwidth}
        \centering
        \includegraphics[width=\linewidth, trim=0 0 0 0, clip]{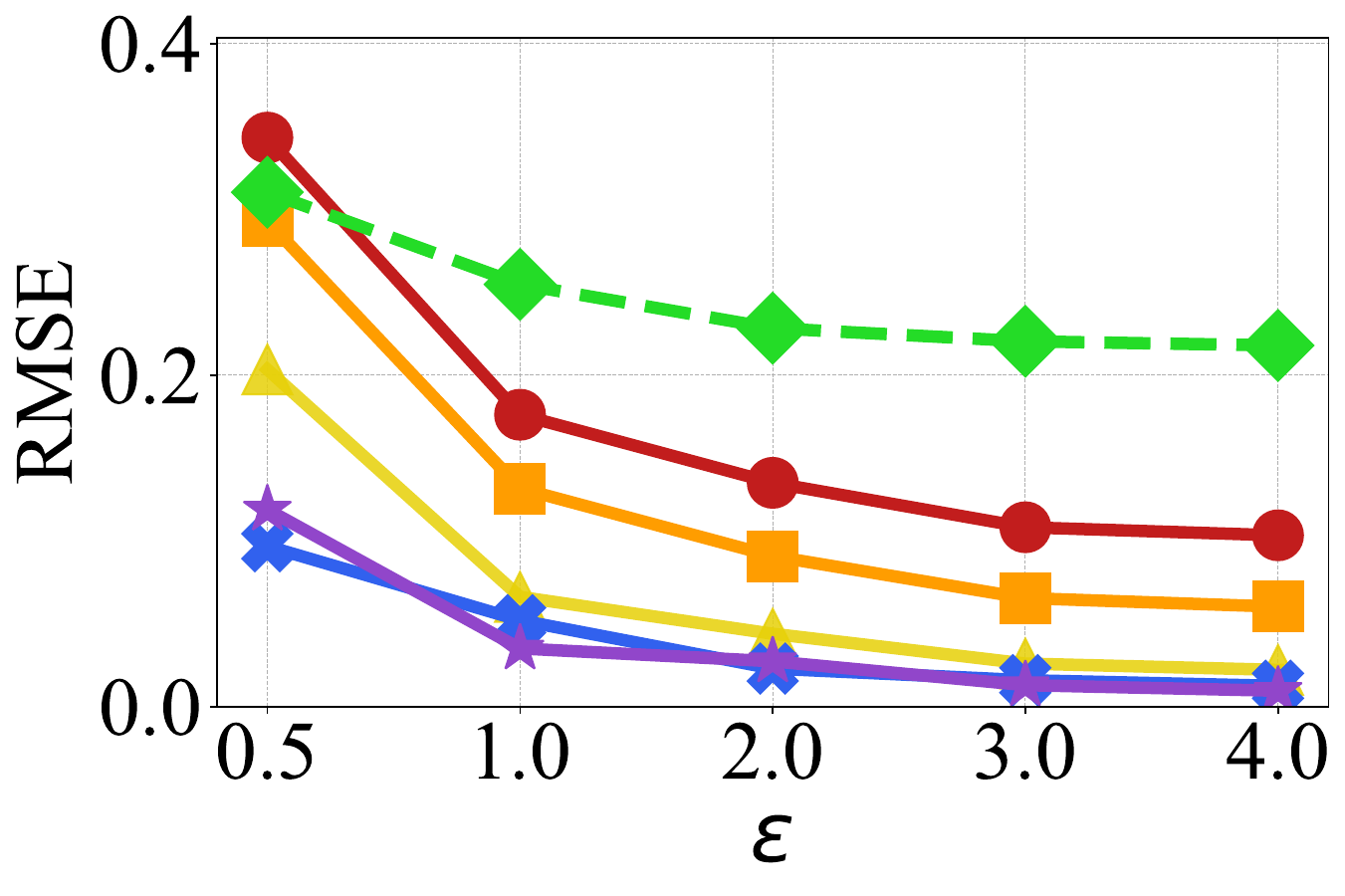}
        \caption{HPC Voltage ($\tau$)}
        \label{fig:tau-hpc}
    \end{subfigure}
\end{subfigure}
\caption{Ablation studies for hyperparameters. RMSE is measured by varying one hyperparameter at a time, comparing the ABC method (solid lines) with the baseline (dashed line). The results demonstrate the robustness of ABC, as it consistently outperforms the baseline across a wide range of settings for all tested hyperparameters.}
\label{fig:ablation-all}
\end{figure*}

\begin{figure}[ht!]
\centering
\begin{subfigure}{0.99\columnwidth}
    \centering
    \includegraphics[width=0.9\linewidth]{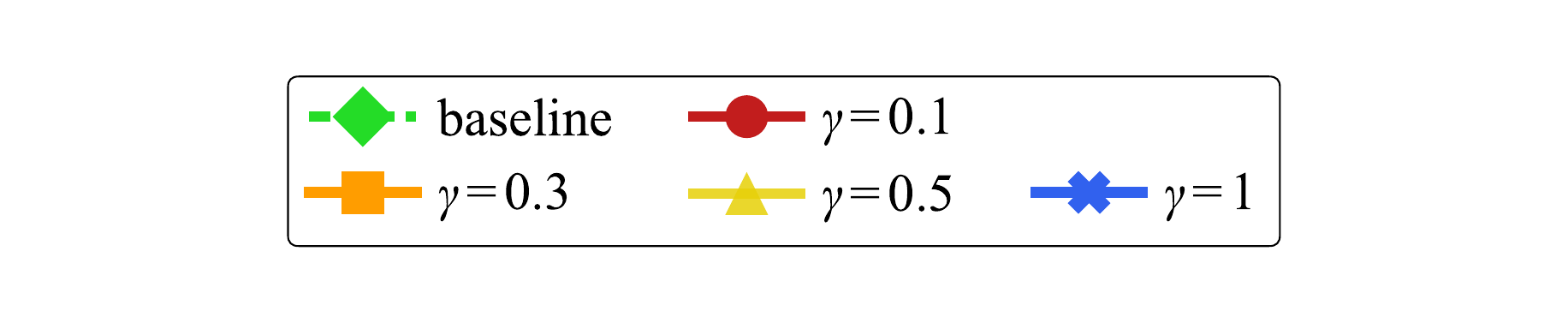}
    \begin{subfigure}{0.49\columnwidth}
        \centering
        \includegraphics[width=\linewidth, trim=0 0 0 0, clip]{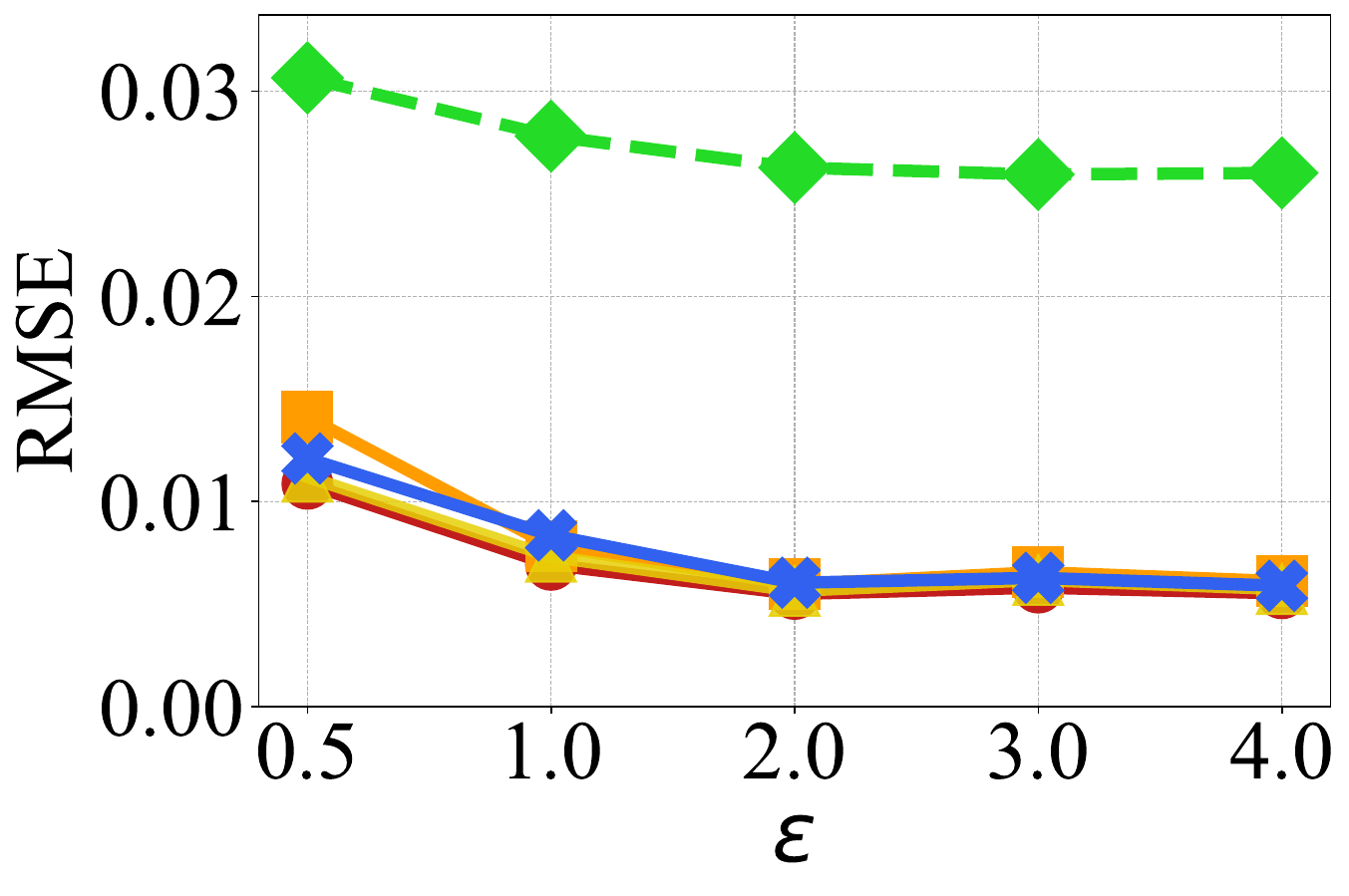}
        \caption{Truncated Normal}
    \end{subfigure}
    \hfill
    \begin{subfigure}{0.49\columnwidth}
        \centering
        \includegraphics[width=\linewidth, trim=0 0 0 0, clip]{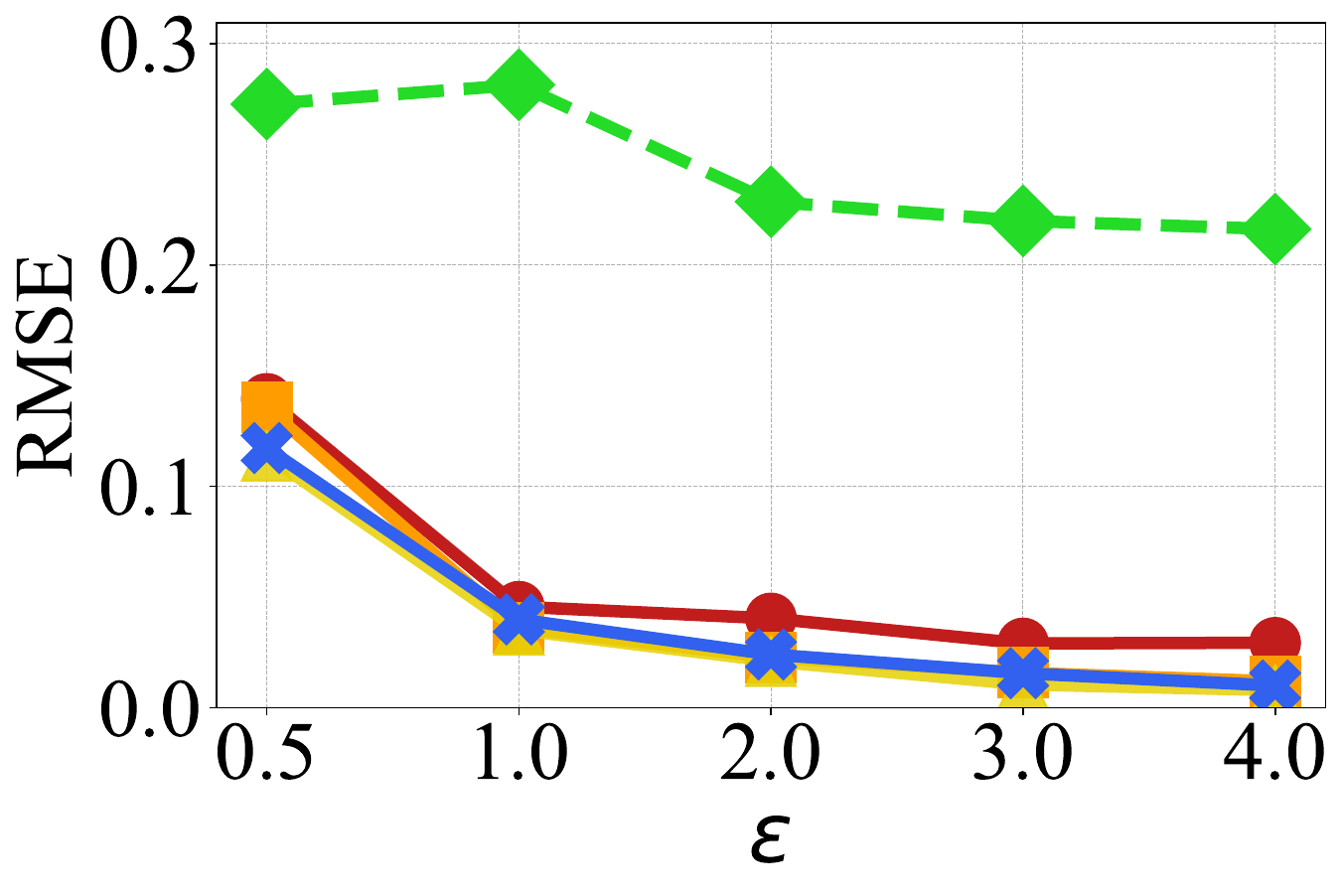}
        \caption{HPC Voltage}
    \end{subfigure}
\end{subfigure}
\caption{Comparison of RMSE between our proposed mechanism under varying non-IID conditions. The results demonstrate that our mechanism achieves significantly lower RMSE compared to the baseline, even in highly non-IID scenarios.}
\label{fig:dirichlet}
\end{figure}

\subsection{Ablation Studies on Hyperparameters} \label{sec:abl}
As the ABC method introduces many hyperparameters, we conduct ablation studies to analyze the robustness to its hyperparameters. Due to page limits and experimental complexity, we limit this study to the two datasets: Truncated Normal (a synthetic dataset) and HPC Voltage (the largest real-world dataset in our experiments). These studies vary one parameter at a time while keeping others fixed to isolate its impact on performance, measured by RMSE using the PM-SUB mechanism.

\subsubsection{Impact of Target Clipping Ratio}
The target clipping ratio, $\alpha$, allows the data collector to control the trade-off between bias and variance. We conducted an ablation study by varying $\alpha$ over the set \{$0, 0.025, 0.05, 0.075, 0.1$\}. The goal of our method is not simply to achieve zero clipping, which would match the true data range $[l^*, r^*]$. For both datasets, setting $\alpha$ to $0$ leads to high variance due to the wide bounds while it still have better performance than baselines. Notably, our ABC method consistently outperforms the baseline for all tested values of $\alpha$. Based on the results, we recommend setting $\alpha$ to a small value, such as between 0.01 and 0.05.

\subsubsection{Impact of Learning Rate}
The learning rate, $\eta$, determines the step size for updating the bounds $l_t$ and $r_t$. To investigate its impact, we varied $\eta$ across the values \{$0.1, 0.2, 0.3, 0.4, 0.5$\}, shown in Figure \ref{fig:lr-tn} and \ref{fig:lr-hpc}.

As shown in the figures, our ABC method consistently outperforms the baseline for all tested values of $\eta$. This shows that the performance gain comes from the adaptive bounding strategy itself, rather than from a single, and our ABC method is not overly sensitive to the choice of learning rate.

\subsubsection{Impact of Number of Rounds}
The hyperparameter $T$ determines the number of rounds into which the total data collection process is divided. For a fixed total population size, $T$ is inversely related to the batch size per round: a larger number of rounds implies smaller data batches per round, while fewer rounds result in larger batches. A larger number of rounds implies more frequent bound adjustments with smaller data batches per round, while a smaller number of rounds means fewer updates with larger batches. To analyze the sensitivity of our ABC method to this parameter, we conducted an ablation study varying the number of rounds across the set \{$10, 20, 30, 40, 50$\}. The results in Figure \ref{fig:step-tn} and \ref{fig:step-hpc} show that ABC outperforms the baseline for all tested values of $T$ except for $\epsilon=0.5$, where it requires smaller rounds due to the noise.

\subsubsection{Impact of Privacy Budget Ratio}
The total privacy budget $\epsilon$ is split by the hyperparameter $\beta$ $(0 < \beta < 1)$. A portion, $\epsilon \cdot \beta$, is used to perturb the numerical data, while the remaining $\epsilon \cdot (1 - \beta)$ protects the clipping signal. A larger $\beta$ enhances data utility, while a smaller $\beta$ improves the accuracy of the clipping estimation for boundary adjustments.

To analyze the impact of this trade-off, we conducted an ablation study by varying $\beta$ across the set \{$0.5, 0.6, 0.7, 0.8, 0.9$\}. The results, shown in Figure \ref{fig:beta-tn} and \ref{fig:beta-hpc}, demonstrate that our ABC method consistently outperforms the baseline for all tested values of $\beta$. The key insight from this experiment is the robustness of our method to the choice of $\beta$. Although there are minor performance variations, the RMSE remains low and stable across a wide range of $\beta$ values. While our method demonstrates robustness to the choice of $\beta$, we recommend a value between 0.7 and 0.8 since the numerical value requires more privacy budget than clipping status.

\subsubsection{Impact of Stability Parameter}
The parameter $\zeta$ is incorporated into our update rule as a small constant to ensure numerical stability. Its primary function is to prevent division-by-zero errors in scenarios where the estimated proportion of non-clipped data is close to zero. Although it is important for robustness, $\zeta$ is designed to have no material impact on the performance of the proposed method.

To confirm this, we performed an ablation study by varying the value of $\zeta$. The results show that the estimation accuracy is insensitive to the choice of $\zeta$, as shown in Figure \ref{fig:tau-tn}, \ref{fig:tau-hpc}. The RMSE remains flat across the entire range of values tested for both datasets. This finding empirically validates that $\zeta$ successfully fulfills its role as a stabilizer without influencing the core dynamics of the ABC method. Its value can be set to small positive constant without consideration for performance tuning, confirming the robustness of our update rule.

\subsubsection{Impact of Amplification Parameter}
As established in Lemma \ref{lem:rate}, this parameter controls the convergence behavior of our algorithm. When $\tau=1$, the algorithm yields linear convergence. When $\tau<1$, it achieves superlinear convergence by amplifying small errors. However, when $\tau>1$, it results in sublinear convergence by dampening error signals. We conducted an ablation study with $\tau \in \{1/3,1/2,1,2,3\}$, and the results are presented in Figure \ref{fig:tau-tn} and \ref{fig:tau-hpc}. The experimental results confirm our theoretical predictions, showing that settings with $\tau < 1$ achieve the lowest RMSE, followed by $\tau=1$. While settings with $\tau > 1$ have the highest error among our configurations, they still significantly outperform the baseline. 

\subsection{Robustness to Non-IID Data}
In practical scenarios, data often arrives with temporal correlations or distributional shifts, violating the Independent and Identically Distributed (IID) assumption. To evaluate the robustness against such irregularities, we conducted experiments using syntheticd non-IID data sequences. 

We simulated non-IID conditions by reordering the dataset based on a Dirichlet distribution. Specifically, the original data was first sorted and partitioned into buckets. We then assigned these buckets to different data batches using a Dirichlet distribution with a concentration parameter $\gamma \in \{0.1, 0.3, 0.5, 1.0\}$.

Figure \ref{fig:dirichlet} presents the RMSE comparison between our method and the baseline across varying $\gamma$ values. The results demonstrate that our ABC mechanism maintains significantly lower RMSE compared to the baseline, even in highly heterogeneous settings ($\gamma=0.1$). This confirms that our adaptive bound estimation strategy effectively prevents the estimation error from diverging even when the input data are non-IID.

\section{Conclusion}

In this paper, we addressed a fundamental limitation of LDP mechanisms for numerical data collection: the requirement of a predefined data domain. This requirement is often impractical in real-world scenarios where no prior knowledge of the data is available, leading to a challenging trade-off between clipping-induced bias and noise-induced variance.

To solve this problem, we proposed the ABC method, a novel and adaptive framework. Our approach enables the server to dynamically learn an appropriate data domain by collecting a privatized signal from users about whether their data was clipped. The core of our method is a principled update rule, grounded in control and information theory, that is proven to converge by Lyapunov stability analysis.

Our extensive experiments demonstrated that the ABC method significantly improves the estimation accuracy of various existing LDP mechanisms across diverse datasets and privacy settings. We also showed that our method is robust to its hyperparameters and intial scales, making it a practical and effective solution. As an orthogonal enhancement, the ABC method can be integrated with current numerical LDP mechanisms, making them more applicable to real-world data collection scenarios where the data domain is unknown.


\clearpage

\section*{AI-Generated Content Acknowledgement}
The authors utilized Google's Gemini 2.5 Pro model for two primary purposes. First, AI-based code generation tools were used to assist in creating Python scripts for plotting the figures in the Experiments section. Second, AI was employed to proofread and refine the language throughout the paper for improved clarity and grammatical correctness. All AI-generated outputs, including code and text suggestions, were reviewed and edited by the authors.

\bibliographystyle{plain}
\bibliography{main}

\end{document}